\documentclass[longbibliography]{amsart}
\usepackage[foot]{amsaddr}
\usepackage[unicode=true,pdfusetitle, bookmarks=true,bookmarksnumbered=false,bookmarksopen=false, breaklinks=false,pdfborder={0 0 0},backref=false,colorlinks=false,pagebackref]{hyperref}
\usepackage{doi}
\usepackage{cite}
\usepackage{fontawesome}
\usepackage{cleveref}
\usepackage{orcidlink}

\hypersetup{colorlinks,linkcolor=myurlcolor,citecolor=myurlcolor,urlcolor=myurlcolor}
\usepackage{graphics,epstopdf,graphicx,amsthm,amsmath,amssymb,mathptmx,braket,colortbl,color,bm,framed,mathrsfs}
\usepackage[up]{subfigure}
\usepackage{tikz}
\usepackage{multirow}
\usepackage{enumerate}
\usepackage[inline]{enumitem}
\usepackage{float}
\usepackage{cite}

\renewcommand{\backref}[1]{}

\renewcommand{\backrefalt}[4]{%
\ifcase #1 %
\or
[p.\ #2]%
\else
[pp.\ #2]%
\fi}

\usepackage{booktabs} 
\let\hbar\undefined
\usepackage{fouriernc} 

\DeclareMathOperator*{\argmax}{arg\,max}

\definecolor{myurlcolor}{rgb}{0,0,0.9}

\newcommand{\ep}[1]{\langle #1 \rangle}

\DeclareMathOperator{\trace}{Tr}
\DeclareMathOperator*{\Expect}{\mathop{\mathbb{E}}}
\newcommand{\Ptr}[2]{\trace_{#1}\Pa{#2}}
\newcommand{\Tr}[1]{\Ptr{}{#1}}

\newcommand{\Pa}[1]{\left[#1\right]}

\theoremstyle{plain}
\newtheorem{thm}{Theorem}
\newtheorem{lem}[thm]{Lemma}

\newtheorem{cor}[thm]{Corollary}

\theoremstyle{definition}

\newtheorem{Rem}[thm]{Remark}

\newtheorem{Examp}[thm]{Example}
\usepackage{qcircuit}
\usepackage{pifont}

\newcommand*{\myproofname}{Proof}

\newcommand{\CDD}{\mathcal D}

\newcommand{\CEE}{\mathcal E}

\newcommand{\BFF}{\mathbb F}
\newcommand{\CFF}{\mathcal F}

\newcommand{\Buu}{\textbf{u}}

\newcommand{\Bxx}{\textbf{x}}

\newcommand{\be}{\begin{equation}}
\newcommand{\ee}{\end{equation}}

\renewcommand{\ge}{\geqslant}
\renewcommand{\geq}{\geqslant}
\renewcommand{\leq}{\leqslant}
\renewcommand{\le}{\leqslant}

\usepackage[inline]{enumitem}
\def\maxlinsat{\textsc{Max-LinSAT}}
\def\maxxorsat{\textsc{Max-XORSAT}}
\def\OPI{\textsc{OPI}}

\newcommand{\mm}{\parbox{1.8cm}{\centering \footnotesize matrix\\multiplication}
}

\newcommand{\sd}{\parbox{1.4cm}{\centering \footnotesize syndrome\\decoding}
}
\renewcommand{\le}{\leqslant}

\newcommand{\qinc}{Quantum Innovation Centre (Q.InC), Agency for Science, Technology and Research (A*STAR), 2 Fusionopolis Way, Innovis \#08-03, Singapore 138634, Republic of Singapore}

\newcommand{\sutd}{Science, Mathematics and Technology Cluster, Singapore University of Technology and Design, 8 Somapah Road, Singapore 487372, Republic of Singapore}

\newcommand{\ihpc}{Institute of High Performance Computing (IHPC), Agency for Science, Technology and Research (A*STAR), 1 Fusionopolis Way, \#16-16 Connexis, Singapore 138632, Republic of Singapore}

\DeclareMathAlphabet{\mathcal}{OMS}{cmsy}{m}{n}

\makeatother
\title{Decoded Quantum Interferometry Under Noise}

\author{Kaifeng Bu$^{1,2}$\,\orcidlink{0000-0001-8419-8736}}
\email{\href{mailto:bu.115@osu.edu}{bu.115@osu.edu} (K.Bu)}
 \address{$^1$\textnormal{Department of Mathematics, The Ohio State University, Columbus, Ohio 43210, USA}}
\address{$^2$\textnormal{Department of Physics, Harvard University, Cambridge, Massachusetts 02138, USA}}

\author{Weichen Gu$^1$}
\email{\href{mailto:gu.1213@osu.edu}{gu.1213@osu.edu} (W.Gu)}

\author{Dax Enshan Koh$^{3,4,5}$\,\orcidlink{0000-0002-8968-591X}}
\email{\href{mailto:dax_koh@ihpc.a-star.edu.sg}{dax\_koh@ihpc.a-star.edu.sg} (D.E.Koh)}
\address{$^3$\textnormal{\qinc}}
\address{$^4$\textnormal{\ihpc}}
\address{$^5$\textnormal{\sutd}}

\author{Xiang Li$^1$}
\email{\href{mailto:li.15497@osu.edu}{li.15497@osu.edu} (X.Li)}

\begin{document}

\begin{abstract}
Decoded Quantum Interferometry (DQI) is a recently proposed quantum optimization algorithm that exploits sparsity in the Fourier spectrum of objective functions, with the potential for exponential speedups over classical algorithms on suitably structured problems. While highly promising in idealized settings, its resilience to noise has until now been largely unexplored. To address this, we conduct a rigorous analysis of DQI under noise, focusing on local depolarizing noise. For the maximum linear satisfiability problem, we prove that, in the presence of noise, performance is governed by a noise-weighted sparsity parameter of the instance matrix, with solution quality decaying exponentially as sparsity decreases. We demonstrate this decay through numerical simulations on two special cases: the Optimal Polynomial Intersection problem and the Maximum XOR Satisfiability problem. The Fourier-analytic methods we develop can be readily adapted to other classes of random Pauli noise, making our framework applicable to a broad range of noisy quantum settings and offering guidance on preserving DQI's potential quantum advantage under realistic noise.
\end{abstract}

\maketitle

\setcounter{tocdepth}{1}
\tableofcontents

\section{Introduction}

Quantum optimization \cite{abbas2024challenges}---the task of using quantum algorithms to find optimal or near-optimal solutions from a space of feasible configurations---has emerged as a prominent approach in the pursuit of practical quantum advantage \cite{leng2025sub,pirnay2024principle,huang2025vast}. Several classes of algorithms within this approach have been extensively explored, including Grover's algorithm \cite{grover1996fast}, which offers a quadratic speedup for unstructured search over solution spaces; quantum adiabatic algorithms \cite{farhi2000quantum,albash2018adiabatic}, which gradually evolve a Hamiltonian whose ground state at the end of the evolution encodes the optimal solution; and variational methods \cite{moll2018quantum,cerezo2021variational} such as the quantum approximate optimization algorithm (QAOA) \cite{farhi2014quantum} and its low-depth variants \cite{herrman2022multi,vijendran2023expressive,shi2022multiangle,zhao2025symmetry}, which encode the cost function into a problem Hamiltonian and seek to approximate its ground state by minimizing the Hamiltonian's expectation value with respect to a parameterized quantum state optimized via a classical feedback loop \cite{pellow2024effect}.

Despite their promise, these approaches face critical challenges. Grover's speedup often vanishes once the oracle's internal structure is accessible classically \cite{stoudenmire2024opening}; adiabatic methods require evolution times that scale inversely with the minimum spectral gap, yielding exponential runtimes when the gap is exponentially small \cite{farhi2011quantum}; and variational algorithms lack general performance guarantees \cite{zhou2020quantum}, suffer from barren plateaus \cite{mcclean2018barren,wang2021noise,cerezo2023does, larocca2025barren} and reachability deficits \cite{akshay2020reachability}, and incur significant classical tuning overhead \cite{bittel2021training,rajakumar2024trainability}.

Decoded Quantum Interferometry (DQI), recently introduced by Jordan et al.~\cite{jordan2024optimization}, offers a fresh, non-variational alternative for quantum optimization. It harnesses quantum interference as its core resource, using a quantum Fourier transform to concentrate amplitudes on symbol strings associated with large objective values---thereby increasing the likelihood of sampling high-quality solutions. DQI leverages the sparsity that frequently characterizes the Fourier spectra of objective functions for combinatorial optimization problems, and can additionally exploit more intricate spectral structure when present. These features suggest a scalable approach with the potential for exponential speedups in specific classes of problems.

Since its introduction, subsequent work has begun to deepen DQI's theoretical and practical foundations. Patamawisut et al.\ developed explicit quantum circuit constructions for all components of DQI, including a decoder based on reversible Gauss--Jordan elimination using controlled-not and Toffoli gates, and performed a detailed resource analysis (covering depth, gate count, and qubit overhead) validated through simulations on maximum cut (MaxCut) instances with up to 30 qubits \cite{patamawisut2025quantum}. Meanwhile, Chailloux and Tillich improved the DQI-based optimal polynomial interpolation (OPI) algorithm by incorporating the Koetter--Vardy soft decoder for Reed--Solomon codes, broadening the class of structured problems for which DQI may offer advantage \cite{chailloux2025quantum}. More recently, Ralli et al.\  proposed incorporating DQI into self-consistent field (SCF) algorithms, introducing DQI-SCF as a hybrid quantum-classical strategy for optimizing Slater determinants, with potential applications in quantum chemistry workflows
\cite{ralli2025bridging}.

While these studies advance DQI under idealized assumptions, a key open question remains: How resilient is DQI to noise? Imperfections such as decoherence and gate infidelity---especially prevalent on near-term quantum devices \cite{preskill2018quantum,cheng2023noisy,preskill2025beyond}---can distort the interference patterns that DQI relies on to amplify high-quality solutions. Understanding how noise impacts DQI is therefore critical to assessing its practical viability.

In this work, we address this question by rigorously analyzing the performance of DQI under noise, focusing on the case of local depolarizing noise acting on the output state. Our analysis adopts a standard noise model satisfying the GTM (gate-independent, time-stationary, and Markovian) assumptions, where the noisy channel is modeled as an ideal unitary followed by a noise channel. This abstraction---widely used, for example, in shadow tomography \cite{koh2022classical,chen2020robust,Bunpj22,wu2024error-mitigated} and channel estimation \cite{flammia2020efficient,chen2022quantum}---captures key features of realistic noise while enabling analytical tractability.

For concreteness, we work with the maximum linear satisfiability (\maxlinsat) problem over a finite field $\mathbb{F}_p$, where the goal is to satisfy as many linear constraints as possible. Each instance is specified by a matrix $B$ whose rows encode the coefficients of the constraints, and we consider the effect of noise with strength $\epsilon$---the local depolarizing rate acting on each qudit of the output state. We show that the expected number of satisfied constraints after measurement is governed by the noise-weighted sparsity $\tau_1(B, \epsilon)$, a parameter that also determines the associated dual code. 

Our main result establishes that, in the presence of noise, the algorithm's performance decays exponentially as the sparsity of the matrix $B$ decreases, revealing a quantitative link between structural properties of the instance and robustness to noise. We illustrate these findings with numerical simulations on two special cases of \maxlinsat: the Optimal Polynomial Intersection (\OPI) problem and the Maximum XOR Satisfiability (\maxxorsat) problem over $\mathbb{F}_2$ \cite{jordan2024optimization}. In both cases, the results display the expected decay in performance under noise. The Fourier-analytic techniques underlying our analysis extend directly to other classes of random Pauli noise, making the framework readily adaptable to a broad range of noisy quantum scenarios.

The rest of the paper is structured as follows. In \cref{sec:preliminaries}, we review the DQI algorithm and its application to the \maxlinsat{} problem. In \cref{sec:noisy_DQI_with_code_distance_constraints}, we analyze the effects of noise on the behavior of the DQI algorithm under a minimum distance assumption on the underlying code. In \cref{sec:noisy_dqi_without_code_distance_constraints}, we relax this assumption and address the resulting challenges, including the non-orthogonality of certain states and a nonzero decoding failure rate. Finally, in \cref{sec:conclusion}, we summarize our main results on the noise resilience of DQI and outline promising directions for future research, including error mitigation strategies, extensions to other noise models, and comparisons with other quantum optimization algorithms.

\section{Preliminaries}
\label{sec:preliminaries}

We begin by reviewing the Decoded Quantum Interferometry (DQI) algorithm~\cite{jordan2024optimization}, with a focus on its application to the maximum linear satisfiability (\maxlinsat) constraint satisfaction problem.

We start by defining \maxlinsat{} over the field $\mathbb F_p$, where $p$ is prime. An instance consists of a matrix $B \in \mathbb F_p^{m\times n}$ and subsets $F_1,\ldots,F_m \subseteq \mathbb F_p$, and the task is to find an assignment $\mathbf{x^*} \in \mathbb F_p^n$ that satisfies as many constraints $(B\mathbf{x})_i \in F_i$
 as possible, i.e.,
    \begin{align}
    \label{eq:opt_solution}
        \mathbf{x}^* \in \argmax_{\mathbf{x} \in \mathbb F_p^n} \bigg| 
        \big\{ i \in [m]: (B\mathbf{x})_i \in F_i \big\}
        \bigg|.
    \end{align}

This optimization can be expressed equivalently in terms of an objective function $f: \mathbb F_p^n \to \mathbb Z$ defined by
\[f(\Bxx) = \sum_{i=1}^m f_i\left((B\mathbf{x})_i \right)
= \sum_{i=1}^m f_i\left(\sum_{j=1}^n B_{ij} x_j \right),\]
where $f_i: \mathbb F_p \to \{-1,1\}$ is a $\pm 1$-valued indicator function:
\[ f_i(x) = \left\{ 
\begin{aligned}
&1, && \text{ if } x\in F_i;\\
&-1, && \text{ otherwise.}
\end{aligned}\right.\]

The DQI algorithm~\cite{jordan2024optimization} uses the quantum Fourier transform to reduce such optimization problems to decoding problems, with the goal of recovering the optimal solution $\mathbf{x}^* \in \mathbb{F}_p^n$ satisfying \cref{eq:opt_solution}. 
Its core idea is to encode the problem into a quantum state, called the DQI state, which has the form:
\begin{align}
    |P(f)\rangle=\sum_{\mathbf{x} \in \mathbb{F}_p^n} P(f(\mathbf{x}))|\mathbf{x}\rangle,
    \label{eq:DQI_state}
\end{align}
where $P(f)$ is a polynomial of $f$. Measuring this state in the computational basis yields a candidate solution $\mathbf x$. By choosing $P(f)$ appropriately, the measurement outcomes can be biased towards $\mathbf{x}^*$, making it highly probable to find the correct answer in only a few measurements. 

To construct $P(f)$, we first introduce some notation. Let $\omega_p=e^{i 2 \pi / p}$, and assume that the sets $F_1,...,F_m$ all have the same cardinality $r:= |F_i| \in \{1,...,p-1\}$. Define $g_i(x):=\frac{f_i(x)-\bar{f_i}}{\varphi}$,
where $\bar{f_i}:=\frac{1}{p} \sum_{x \in \mathbb{F}_p} f_i(x)$ and $\varphi:=\left(\sum_{y \in \mathbb{F}_p}\left|f_i(y)-\bar{f_i}\right|^2\right)^{1 / 2}$. The Fourier transform of $g_i$ is given by
$
\tilde{g}_i(y)=\frac{1}{\sqrt{p}} \sum_{x \in \mathbb{F}_p} \omega_p^{y x} g_i(x),
$
which is equal to $0$ at $y=0$ and is normalized:
$\sum_{x \in \mathbb{F}_p}\left|g_i(x)\right|^2=\sum_{y \in \mathbb{F}_p}\left|\tilde{g}_i(y)\right|^2=1$.

Let $\mathbf{b}_i$ be the $i$-th row of $B$. For $k\ge 1$, define
\begin{align}\label{47_1}
P^{(k)}\left(g_1\left(\mathbf{b}_1 \cdot \mathbf{x}\right), \ldots, g_m\left(\mathbf{b}_m \cdot \mathbf{x}\right)\right)=\sum_{\substack{i_1, \ldots, i_k \\ \text { distinct }}} \prod_{i \in\left\{i_1, \ldots, i_k\right\}} g_i\left(\mathbf{b}_i \cdot \mathbf{x}\right),
\end{align}
and the corresponding normalized state
\begin{align}\label{P^(k)}
    \left|P^{(k)}\right\rangle = \frac{1}{\sqrt{p^{n-k} \binom{m}{k}}} \sum_{\Bxx \in \BFF_p^n} P^{(k)}\left(g_1\left(\mathbf{b}_1 \cdot \mathbf{x}\right), \ldots, g_m\left(\mathbf{b}_m \cdot \mathbf{x}\right)\right) \ket{\Bxx}.
\end{align} 
Then, the DQI state (\cref{eq:DQI_state}) can be expressed as
\begin{align}\label{51}
|P(f)\rangle=\sum_{k=0}^{l} w_k\left|P^{(k)}\right\rangle,
\end{align}
where $w_0,...,w_l$ are coefficients that satisfy the normalization condition $\sum_k |w_k|^2 =1$.
Further background and derivations can be found in \cref{appen:background}.

A high-level summary of DQI is shown in the circuit diagram of \Cref{fig:dqi_circuit}, where the main steps of matrix multiplication and syndrome decoding are depicted. In the noiseless setting, the subsequent measurements directly yield the solution with high probability. In the noisy setting we consider, however, local depolarizing noise acts on the output state just before measurement, effectively modeling measurement errors. The next section analyzes how such noise impacts DQI's performance.

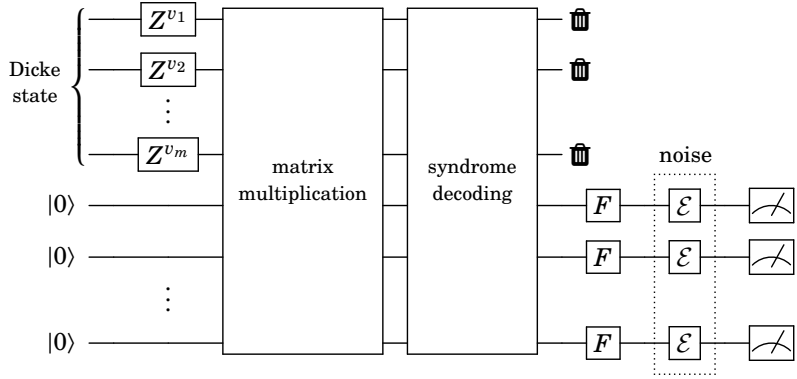
\begin{figure}
\hspace*{0.08\textwidth}
\Qcircuit @C=1em @R=.7em {
& \qw 
& \gate{Z^{v_1}} 
& \multigate{9}{\mm} 
& \multigate{9}{\sd}  
& \qw
& \rstick{\mbox{\hspace{-0.7cm} \faTrash}}
\\
& \qw 
& \gate{Z^{v_2}} 
& \ghost{\mm} 
& \ghost{\sd}
& \qw
& 
\rstick{\mbox{\hspace{-0.7cm} \faTrash}}
\\
&
& \vdots
& 
&
&
\\
&
&
&
&
&
\\
& \qw 
& \gate{Z^{v_m}} 
& \ghost{\mm}
& \ghost{\sd}
& \qw
&
\rstick{\mbox{\hspace{-0.7cm} \faTrash}}
& \mbox{\hspace{1cm} \small noise}
\\
\lstick{\ket 0}
& \qw 
& \qw
& \ghost{\mm}
& \ghost{\sd}
& \qw
& \gate{F}
& \qw
& \gate{\mathcal E}
& \qw
& \meter
\\
\lstick{\ket 0}
& \qw 
& \qw 
& \ghost{\mm}
& \ghost{\sd}
& \qw
& \gate{F}
& \qw
& \gate{\mathcal E}
& \qw
& \meter
\\
&
& \vdots
&
&
&
&
&
&
&
&
\\
&
&
&
&
&
&
&
&
\\
\lstick{\ket 0}
& \qw 
& \qw 
& \ghost{\mm}
& \ghost{\sd}
& \qw 
& \gate{F}
& \qw
& \gate{\mathcal E}
& \qw
& \meter
\inputgroupv{1}{5}{.8em}{.8em}{}
\inputgrouph{2}{5}{0.5em}{\parbox{1cm}{\centering \footnotesize Dicke\\state}}{2.2em}
\gategroup{6}{9}{10}{9}{1.2em}{.}
}

\caption{An example of a quantum circuit for Decoded Quantum Interferometry (DQI), subject to local noise at the output.
}
    \label{fig:dqi_circuit}
\end{figure}

\section{Noisy DQI with Code Distance Constraints}
\label{sec:noisy_DQI_with_code_distance_constraints}

In this section, we examine the effects of noise on the performance of the DQI algorithm, focusing on local depolarizing noise acting on the output state just before measurement. In this setting, the expected number of constraints satisfied depends on the noise level and the sparsity of the instance matrix $B$. The theorem below quantifies this dependence under a minimum code distance assumption on $B$.

\begin{thm}\label{thm:main_1}
Let $f(\mathbf{x})=\sum_{i=1}^m f_i\left(\sum_{j=1}^n B_{i j} x_j\right)$ be a \maxlinsat{} objective function with matrix $B \in \BFF_p^{m \times n}$ for a prime $p$ and positive integers $m$ and $n$ such that $m>n$. 
Suppose that $\left|f_i^{-1}(+1)\right|=r$ for some $r \in\{1, \ldots, p-1\}$. 
Let $P$ be a degree-$l$ polynomial determined by coefficients $w_0,...,w_l$ such that the DQI state $\ket{P(f)}$ satisfies \eqref{51}.
Let $\left\langle s_D^{(m, l)}\right\rangle$ be the expected number of satisfied constraints for the symbol string obtained upon measuring the errored DQI state $\CEE^{\otimes n} (\ket{P(f)}\bra{P(f)})$ in the computational basis, where $\CEE(\rho) =  (1-\varepsilon)\rho + \varepsilon\Tr{\rho}I/p$ denotes the depolarizing channel.
If $2 l+1<d^{\perp}$, where $d^{\perp}$ is the minimum distance of the code $C^{\perp}=\left\{\mathbf{v} \in \mathbb{F}_p^m: B^T \mathbf{v}=\mathbf{0}\right\}$, then
\begin{align}\label{res:main_1}
    \left\langle s_D^{(m, l)}\right\rangle=\frac{m r}{p}+\tau_1(B,\varepsilon) \frac{\sqrt{r(p-r)}}{p} \mathbf{w}^{\dagger} A^{(m, l, d)} \mathbf{w},
\end{align}
where
\begin{align}\label{eq:noise_P}
    \tau_1(B,\varepsilon)=\mathbb{E}_i\tau(B,\varepsilon,i),\quad
    \tau(B,\varepsilon,i) =  (1-\varepsilon)^{|\mathbf{b}_i|},
\end{align}
with $|\mathbf{b}_i|$
denoting the number of non-zero entries of the $i$-th row of matrix $B$,
$\mathbf{w}=\left(w_0, \ldots, w_{l}\right)^T$ is a unit vector and $A^{(m, l, d)}$ is the $(l+1) \times(l+1)$ symmetric tridiagonal matrix
\begin{align}\label{111}
A^{(m, l, d)}=\left[\begin{array}{ccccc}
0 & a_1 & & & \\
a_1 & d & a_2 & & \\
& a_2 & 2 d & \ddots & \\
& & \ddots & & a_{l} \\
& & & a_{l} & l d
\end{array}\right]
\end{align}
with $a_k=\sqrt{k(m-k+1)}$ and $d=\frac{p-2 r}{\sqrt{r(p-r)}}$.
Hence, if the matrix $B$ satisfies the following sparsity condition: 
$L_1\leq |\mathbf{b}_i|\leq L_2, \forall i\in [m]$, then 
\begin{align}\label{res:main_11}
  (1-\varepsilon)^{L_2}\frac{\sqrt{r(p-r)}}{p} \mathbf{w}^{\dagger} A^{(m, l, d)} \mathbf{w}\leq  \left\langle s_D^{(m, l)}\right\rangle- \frac{m r}{p}
   \leq (1-\varepsilon)^{L_1}\frac{\sqrt{r(p-r)}}{p} \mathbf{w}^{\dagger} A^{(m, l, d)} \mathbf{w}.
\end{align}
\end{thm}
When the noise parameter $\epsilon =0$, the theorem above reduces to the result in \cite{jordan2024optimization}.
The proof of the above theorem is presented in Appendix~\ref{appen:A}. 
Based on the results of \eqref{res:main_1} and \eqref{res:main_11}, we find that high sparsity of the matrix $B$---that is, a large proportion of zero entries---is necessary to improve the expected number of satisfied constraints in this noisy case.

\begin{Examp}\label{exam:OPI}
    Consider the Optimal Polynomial Intersection (OPI) problem, an example that highlights the potential quantum speedup of the DQI algorithm on certain structured tasks \cite{jordan2024optimization}. The problem may be stated as follows: given an integer $n<p$ and subsets $F_1,...,F_{p-1} \subseteq \BFF_p$, the task is to find a polynomial $Q\in \BFF_p[y]$ of degree at most $n-1$ that maximizes the function $f_{\mathrm{OPI}}(Q) = |\{ y\in \{1,...,p-1\}: Q(y) \in F_y\}|$, which counts the number of subsets it intersects.

    Note that the OPI problem is a special case of the \maxlinsat{} problem,
    where the corresponding matrix $B=(B_{ij})$ is a $(p-1)\times n$ matrix with entries $B_{ij} = i^{j-1} $. Hence,  $\tau_1(B,\varepsilon) =(1- \varepsilon)^n$,  which has exponential decay.
     Figure \ref{fig:2} shows the exponential decay of the $\tau_1(B,\varepsilon)$ in the OPI problem for local dimension $p=97$.

\begin{figure}[htbp]
    \centering
    \begin{minipage}{0.48\textwidth}
        \centering
        \includegraphics[width=\textwidth]{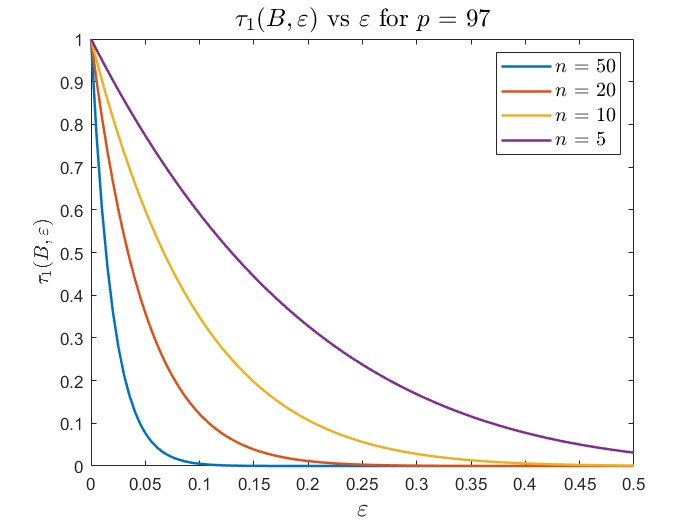}
    \end{minipage}
    
    \caption{  Here is the diagram of  $\tau_1(B,\varepsilon)$ for   the OPI problems with $0\le \varepsilon< 0.5$ and local dimension $p=97$,  as presented in Example \ref{exam:OPI}. }
    \label{fig:2}
\end{figure}

\end{Examp}

\begin{Examp}\label{exam:XORSAT}
    Here we discuss a class of sparse max-XORSAT problems considered in \cite{jordan2024optimization}.
Given a max-XORSAT problem $B\mathbf{x}  = \mathbf{v}$, 
the number of nonzero entries in the $i$-th row of $B$ is denoted as $D_i$, called the degree of the $i$-th constraint.
For convenience, we also denote 
$\kappa_j$ as the fraction of constraints that have degree $j$;
hence, $\sum_{j}  \kappa_j =1$.
By Lemma \ref{250627lem2},
$\tau(B,\varepsilon,i) =  (1-\varepsilon)^{D_i}$,
so 
\begin{align*} 
\tau_1(B,\varepsilon) =  \sum_{j} \kappa_j (1-\epsilon)^j, \quad \text{ and } \quad 
\tau_\infty(B,\varepsilon) =  \max
\left\{(1-\epsilon)^j : \kappa_j >0 \right\}.
\end{align*}
See Figure \ref{250717fig2} for a plot of the behavior of the example in \cite{jordan2024optimization}.
\begin{figure}[htbp]
\centering
    \begin{minipage}{0.48\textwidth}
        \centering
        \includegraphics[width=\textwidth]{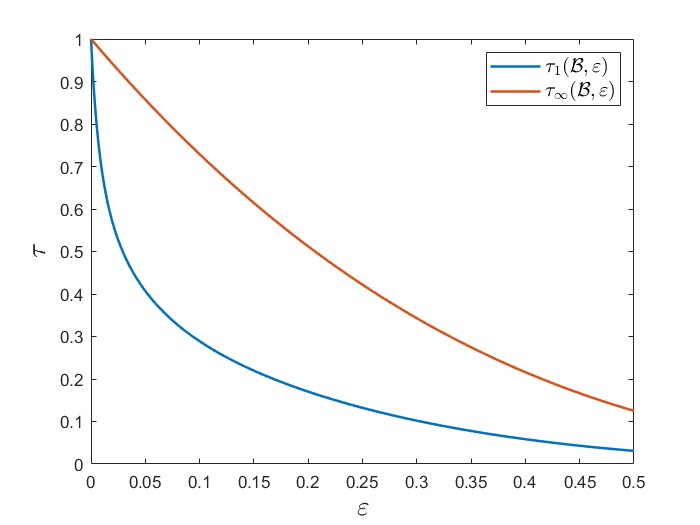}
    \end{minipage}\hfill

\caption{The functions   
$\tau_1(B,\varepsilon)$ and $\tau_\infty(B,\varepsilon)$ for the constraint degree distribution, as presented in Example~\ref{exam:XORSAT}.
    }
    \label{250717fig2}
\end{figure}

\end{Examp}

\begin{Rem}
In the above calculation, we have focused on the effect of the depolarizing channel on the expected number of satisfied constraints in the DQI algorithm. More generally, by applying the same Pauli-basis Fourier analysis~\cite{montanaro2010quantum,BD19,Bucomplexity22,BGJ23a,BGJ23b,BGJ23c,Bu2025quantum,BJPRL25}, our results extend to any random Pauli channel.
\end{Rem}

Therefore, to optimize the performance of the DQI algorithm, we choose $\mathbf{w}$  to be the principal eigenvector of $A^{(m,l,d)}$.
The following lemma from \cite{jordan2024optimization} provides an estimate of the largest eigenvalue of 
 $A^{(m,l,d)}$.

\begin{lem} [\mbox{Jordan et al.~\cite[Lemma 9.3]{jordan2024optimization}}] \label{lem9.3}Let $\lambda_{\max }^{(m, l)}$ denote the maximum eigenvalue of the symmetric tridiagonal matrix $A^{(m, l, d)}$ defined in \eqref{111}. If $l \leq m / 2$ and $d \geq-\frac{m-2 l}{\sqrt{l(m-l)}}$, then
$$
\lim _{\substack{m, l \rightarrow \infty \\ l / m=\mu}} \frac{\lambda_{\max }^{(m, l, d)}}{m}=\mu d+2 \sqrt{\mu(1-\mu)},
$$
where the limit is taken in the regime where both $m$ and $l$ tend to infinity, with the ratio $\mu=l / m$ fixed.
\end{lem}

\begin{cor}
Under the same assumption as Theorem~\ref{thm:main_1}, 
in the limit as $m\to \infty$,
with $l/m$ fixed such that $ \frac lm \ge 1- \frac rp$,
the optimal choice of degree-$l$ polynomial $P$ to maximize $\left\langle s_D^{(m, l)}\right\rangle$  yields
\begin{align}
\lim _{\substack{m, l \rightarrow \infty \\ l / m=\mu}} \frac{\left\langle s_D^{(m, l)}\right\rangle_{\mathrm{opt}}}{m}
=  \frac{r}{p}+ 
\tau_1(B,\varepsilon)\left(\mu - 2   \mu \frac{r}{p} +  2  
\sqrt{\frac{r}{p} \left( 1 - \frac{r}{p}\right)}
\sqrt{\mu(1-\mu)}\right).
\end{align}
Hence, if the matrix $B$ satisfies the following sparsity condition: 
$L_1\leq |B_i|\leq L_2, \forall i\in [m]$, then 
\begin{align*}
  &(1-\varepsilon)^{L_2}\left(\mu - 2   \mu \frac{r}{p} +  2  
\sqrt{\frac{r}{p} \left( 1 - \frac{r}{p}\right)}
\sqrt{\mu(1-\mu)}\right)\\
&\leq \lim _{\substack{m, l \rightarrow \infty \\ l / m=\mu}} \frac{\left\langle s_D^{(m, l)}\right\rangle_{\mathrm{opt}}}{m}
-  \frac{r}{p}\\
   &\leq (1-\varepsilon)^{L_1}\left(\mu - 2   \mu \frac{r}{p} +  2  
\sqrt{\frac{r}{p} \left( 1 - \frac{r}{p}\right)}
\sqrt{\mu(1-\mu)}\right).
\end{align*}
\end{cor}
\begin{proof}
Due to the fact that $d = \frac{p-2r}{\sqrt{r(p-r)}} = \sqrt{\frac{p-r}{r}} - \sqrt{\frac{r}{p-r}} 
\ge \sqrt{\frac{l}{m-l}} - \sqrt{\frac{m-l}{l}} = - \frac{m-2l}{\sqrt{l(m-l)}}$,
the condition specified in Lemma~\ref{lem9.3}  is satisfied.
Hence,
by Theorem \ref{thm:main_1}, we have
\begin{align*}
\lim _{\substack{m, l \rightarrow \infty \\ l / m=\mu}} \frac{\left\langle s_D^{(m, l)}\right\rangle_{\mathrm{opt}}}{m}
&= \frac{r}{p}+\tau_1(B,\varepsilon) 
\sqrt{\frac{r}{p} \left( 1 - \frac{r}{p}\right)} 
\lim _{\substack{m, l \rightarrow \infty \\ l / m=\mu}} \frac{\lambda_{\max }^{(m, l, d)}}{m}\\
&= \frac{r}{p}+\tau_1(B,\varepsilon) 
\sqrt{\frac{r}{p} \left( 1 - \frac{r}{p}\right)} 
(\mu d+2 \sqrt{\mu(1-\mu)})\\
&= \frac{r}{p}+ 
\tau_1(B,\varepsilon)\left(\mu - 2   \mu \frac{r}{p} +  2  
\sqrt{\frac{r}{p} \left( 1 - \frac{r}{p}\right)}
\sqrt{\mu(1-\mu)}\right),
\end{align*} 
where the last line follows from the fact that $d = \frac{p-2r}{\sqrt{r(p-r)}}$.
\end{proof}

\section{Noisy DQI without Code Distance Constraints}
\label{sec:noisy_dqi_without_code_distance_constraints}

In this section, we remove the assumption that the 
minimum distance satisfies $2l+1<d^\perp$. 
This relaxation introduces two problems that need consideration. First, the states  
$ \left|\widetilde{P}^{(0)}\right\rangle, \ldots,$ $\left|\widetilde{P}^{(l)}\right\rangle$ 
are no longer orthogonal to each other.  The general relation of the norms of these states will be discussed in
Lemma \ref{norm of FT of P(f)}.
Second, when preparing the DQI state, the decoding process may have a nonzero failure rate, preventing the exact realization of the ideal DQI state. In this section, we focus on this problem under depolarizing noise.


We assume that the imperfect decoder partitions the set $\BFF_p^m$ of errors  into $\BFF_p^m = \CDD\cup \CFF$, where $\CDD$ denotes the set of errors $\mathbf{y}$ correctly identified by the decoder based on its syndrome $B^T \mathbf{y}$,
and $\CFF$ denotes the set of errors misidentified.
When restricted to the set $E_k$ of all errors with Hamming weight $k$,
we denote $\CDD_k = \CDD \cap E_k$ and $\CFF_k = \CFF \cap E_k$,
then $E_k = \CDD_k \cup \CFF_k$.
The quantum state after the error decoding step of the DQI algorithm using an imperfect decoder is
\begin{align}
\sum_{k=0}^{l} \frac{w_k}{\sqrt{\binom{m}{k}}} 
\left( 
\sum_{\substack{\mathbf{y} \in \CDD_k}}
\tilde{g}(\mathbf{y})
\ket{\mathbf{0}}\left|B^T \mathbf{y}\right\rangle +
\sum_{\substack{\mathbf{y} \in \CFF_k}}
\tilde{g}(\mathbf{y})
\ket{\mathbf{y} \oplus \mathbf{y}'}\left|B^T \mathbf{y}\right\rangle 
\right),
\end{align}
where by $\mathbf{y}' $ we denote the error identified by the decoder based on the syndrome  $B^T \mathbf{y}$ and $\mathbf{y}' \neq \mathbf{y} $.
Then the DQI algorithm postselect on the register being $\ket{\mathbf{0}}$,
and get the following unnormalized state
\begin{align}\label{171}
\ket{ \widetilde{P}_\CDD (f)} := \sum_{k=0}^{l} \frac{w_k}{\sqrt{\binom{m}{k}}}  
\sum_{\substack{\mathbf{y} \in \CDD_k}}
\widetilde{g}(\mathbf{y})
\left|B^T \mathbf{y}\right\rangle .
\end{align}
Then the DQI algorithm provides an output $\mathbf{x}$ by measuring the state $\ket{ {P}_\CDD (f)}$ in the computational basis,
where $\ket{ {P}_\CDD (f)}$ is the inverse quantum Fourier transform of $\ket{ \widetilde{P}_\CDD (f)}$, i.e., $\ket{ \widetilde {P}_\CDD (f)} = F^{\otimes n} \ket{  {P}_\CDD (f)}$ for $F_{i,j} = \omega_{p}^{ij} /\sqrt p$, $i,j =0,...,p-1$.

To quantify the failure rate of the decoder (for a given \maxlinsat{} problem), for each Hamming weight $k$ we define
\begin{align}\label{250610eq1}
\gamma_k :=\frac{|\CFF_k|}{|E_k|} = \frac{|\CFF_k|}{(p-1)^k\binom{m}{k}}
\end{align}
and $\gamma_{\max} := \max_{0\le k\le l} \gamma_k$.
In particular,
when $p=2$,
$\gamma_k = |\CFF_k|/\binom{m}{k}$.

Now, let us estimate the  expected number of satisfied constraints for the symbol string obtained upon measuring the errored imperfect DQI state $\CEE^{\otimes n} (\ket{P_\CDD(f)}\bra{P_\CDD(f)})$ in the computational basis. Our next lemma gives an expression for the square norm of the noisy DQI state.


\begin{lem}\label{lem10.4}
The squared norm of $ \ket{ {P}_\CDD (f)}$ is
\[\ep{ {P}_\CDD (f) | {P}_\CDD (f)} =
\sum_{k=0}^{l} \frac{|w_k|^2}{ {\binom{m}{k}}}  
\sum_{\substack{\mathbf{y} \in \CDD_k}}
|\widetilde{g}(\mathbf{y})|^2
.\]
\end{lem}
\begin{proof}
Since the decoder can identify any error $\mathbf{y}\in \CDD$ only based on its syndrome $B^T \mathbf{y}$,
the syndromes $\ket{B^T \mathbf{y}}$ must be distinct for $\mathbf{y}\in \CDD$, and therefore $\ket{B^T \mathbf{y}}$ are orthogonal states.
Hence,  by the equation \eqref{171},  we have
\begin{align}
    \ep{ \widetilde{P}_\CDD (f) |\widetilde{P}_\CDD (f)}= \sum_{k=0}^{l} \frac{|w_k|^2}{ {\binom{m}{k}}}  
\sum_{\substack{\mathbf{y} \in \CDD_k}}
|\widetilde{g}(\mathbf{y})|^2.
\end{align}
Combining this with \cref{250610eq1} yields the stated result.
\end{proof}

For simplicity,
in the following we will consider the special case where $p=2$ and $r=1$. Note that the noiseless case with $p=2$ and $r=1$
 was previously considered in \cite{jordan2024optimization}.
In this case, the \maxlinsat{} problem reduces to the \maxxorsat{} problem, which may be stated as follows: given a matrix $B\in \mathrm{M}_{m\times n}(\BFF_2)$ and a vector $\mathbf{v}\in \BFF_2^m$, find an $n$-bit string $\mathbf{x}\in \BFF_2^n$ satisfying as many as possible of the $m$ linear equations modulo 2, $B\mathbf{x} =  \mathbf{v}$.

\begin{thm}\label{thm:main_3}
Let $f(\mathbf{x})=\sum_{i=1}^m f_i\left(\sum_{j=1}^n B_{i j} x_j\right)$ be a \maxlinsat{} objective function with matrix $B \in \BFF_2^{m \times n}$ for positive integers $m$ and $n$ such that $m>n$. 
Suppose that $\left|f_i^{-1}(+1)\right|=1$ for every $i$.
Let 
$P(f) $ be the degree-$l$ polynomial determined by coefficients $w_0,...,w_l$ such that the perfect DQI state $\ket{P(f)}$ satisfies \eqref{51} (note that $\ket{P(f)}$ is not normalized).
Let $\left\langle s_D^{(m, l)}\right\rangle$ denote the expected number of satisfied constraints obtained by measuring, in the computational basis, the errored imperfect DQI state $\CEE^{\otimes n} (\ket{P_\CDD(f)}\bra{P_\CDD(f)})$. Suppose the sets $F_1,...,F_m$ are chosen independently uniformly at random from  $\{\{0\},\{1\}\}$.
Then,
$$
\Expect_{F_1,...,F_m} \left\langle s_D^{(m, l)}\right\rangle\ge \frac m2 + \frac12\tau_1(B,\varepsilon)  
\frac{\mathbf{w}^{\dagger} {A}^{(m, l,0)} \mathbf{w}}{\|  \mathbf{w}\|^2} - \tau_\infty(B,\varepsilon)
\frac{ (m+1)  \gamma_{\max} }{1-\gamma_{\max} }
 ,
$$
where ${A}^{(m, l,0)}$ is the tridiagonal matrix defined in \eqref{111}, $\tau_1(B,\varepsilon)=\mathbb{E}_i\tau(B,\varepsilon,i)$, and
$\tau_\infty(B,\varepsilon)=\max_i\tau(B,\varepsilon,i)$, with $\tau(B,\varepsilon,i)$ defined in 
\eqref{eq:noise_P}.
\end{thm}
The proof of this theorem is provided in Appendix~\ref{appen:B}.

\begin{cor}\label{thm10.1}
Under the same assumptions as Theorem~\ref{thm:main_3},
assume that $l\le m/2$ and choose $\mathbf{w}$ to be the principal eigenvector of ${A}^{(m, l,0)}$. Then, we have
\begin{align*}
\lim_{\substack{m\rightarrow \infty\\ l/m=\mu}} \frac 1m \Expect_{F_1,...,F_m} \left\langle s_D^{(m, l)}\right\rangle \ge \frac 12 + \frac12 \tau_1(B,\varepsilon)   \sqrt{\frac{l}{m}\left(1- \frac lm \right)} - \tau_\infty(B,\varepsilon) \frac{  \gamma_{\max} }{1-\gamma_{\max} } .
\end{align*}
\end{cor}
\begin{proof}
This follows from Theorem~\ref{thm:main_3} and Lemma 9.3 in \cite{jordan2024optimization} (or Lemma \ref{lem9.3}).
\end{proof}

\section{Conclusion}
\label{sec:conclusion}
We have investigated the performance of the DQI algorithm in the presence of noise, focusing specifically on depolarizing noise as a representative and analytically tractable model. Our analysis reveals that the expected number of satisfied constraints decreases exponentially with a noise-weighted sparsity measure of the problem's matrix $B$. This dependence uncovers a fundamental sensitivity of DQI to the structural properties of the optimization instance, providing valuable guidance for selecting appropriate optimization strategies in noisy settings. Moreover, our Fourier-based analytical framework applies more broadly: similar results hold for general random Pauli noise channels, enabling extensions to a wider class of physically relevant noise models.

Beyond the findings presented in this work, several important questions warrant further investigation. First, exploring error mitigation techniques or alternative quantum error correction encodings that can preserve DQI's advantages in the presence of noise remains an open and practically motivated challenge.
Second, while we focused on depolarizing noise (and by extension, random Pauli noise) in this work, extending the analysis to other noise models like amplitude damping noise, gate-dependent or non-Markovian noise could provide further insights to the algorithm's robustness. Third, a systematic comparison with the performance of other quantum optimization algorithms under noise---such as QAOA under noisy settings \cite{xue2021effects, marshall2020characterizing}---would help clarify the relative strengths and weaknesses of DQI.





\begin{appendix}

\section{Background about the DQI algorithm}

\label{appen:background}
Let $p$ be a prime, and let $B= (B_{ij})$ be an $m\times n$ matrix over the finite field $\BFF_p$. 
For each $i=1,...,m$, let $F_i\subseteq \BFF_p $ be subsets of $\BFF_p$, which yield a corresponding constraint $\sum_{j=1}^n B_{ij} x_j \in F_i$. The
\textbf{\maxlinsat{} problem} may be stated as follows: find a vector $\Bxx\in \BFF_p^n$ that satisfies as many of these $m$ constraints as possible,
or equivalently, maximize the function
\[f(\Bxx) = \sum_{i=1}^m f_i\left(\sum_{j=1}^n B_{ij} x_j \right),\]
where
\[ f_i(x) = \left\{ 
\begin{aligned}
&1, && \text{ if } x\in F_i;\\
&-1, && \text{ otherwise.}
\end{aligned}\right.\]

The key state in DQI \cite{jordan2024optimization}  is the DQI state $|P(f)\rangle=\sum_{\mathbf{x} \in \mathbb{F}_p^n} P(f(\mathbf{x}))|\mathbf{x}\rangle$, where $P(f)$ is a polynomial of $f$.
The solution $\Bxx$ provided by the DQI algorithm arises from performing a measurement on $\ket{P(f)}$ in the computational basis.

We first introduce some relevant notation and definitions before we discuss the noisy version. Let us denote $\omega_p=e^{i 2 \pi / p}$, and assume that the sets $F_1,...,F_m$ have the same cardinality $r:= |F_i| \in \{1,...,p-1\}$. For simplicity,  let us define functions $g_i$ as follows
\begin{align}\label{250611eq1}
g_i(x):=\frac{f_i(x)-\bar{f_i}}{\varphi},
\end{align}
where $\bar{f_i}:=\frac{1}{p} \sum_{x \in \mathbb{F}_p} f_i(x)$ and $\varphi:=\left(\sum_{y \in \mathbb{F}_p}\left|f_i(y)-\bar{f_i}\right|^2\right)^{1 / 2}$. 
By direct calculation, we have
\begin{align}
\bar{f_i} =\frac{2 r}{p}-1,\quad
\varphi =\sqrt{4 r\left(1-\frac{r}{p}\right)}. \label{250614eq1}
\end{align}
Hence,  for $v\in F_i$, we have
\begin{align}\label{97}
g_i(v) = \frac{1- \bar f_i}{\varphi} = \sqrt{\frac{p-r}{pr}}.
\end{align}

The Fourier transform of $g_i$ is denoted as
\begin{align}\label{37}
\tilde{g}_i(y)=\frac{1}{\sqrt{p}} \sum_{x \in \mathbb{F}_p} \omega_p^{y x} g_i(x),
\end{align}
which is equal to $0$ at $y=0$ and is normalized: $\sum_{x \in \mathbb{F}_p}\left|g_i(x)\right|^2=\sum_{y \in \mathbb{F}_p}\left|\tilde{g}_i(y)\right|^2=1$.

Let $\mathbf{b}_i$ be the $i$-th row in $B$.
For $k\ge 1$, let us  define the polynomials as follows
\begin{align}\label{47}
P^{(k)}\left(g_1\left(\mathbf{b}_1 \cdot \mathbf{x}\right), \ldots, g_m\left(\mathbf{b}_m \cdot \mathbf{x}\right)\right)=\sum_{\substack{i_1, \ldots, i_k \\ \text { distinct }}} \prod_{i \in\left\{i_1, \ldots, i_k\right\}} g_i\left(\mathbf{b}_i \cdot \mathbf{x}\right),
\end{align}
and the corresponding state
\begin{align}
    \left|P^{(k)}\right\rangle = \frac{1}{\sqrt{p^{n-k} \binom{m}{k}}} \sum_{\Bxx \in \BFF_p^n} P^{(k)}\left(g_1\left(\mathbf{b}_1 \cdot \mathbf{x}\right), \ldots, g_m\left(\mathbf{b}_m \cdot \mathbf{x}\right)\right) \ket{\Bxx}.
\end{align}  
The DQI state $|P(f)\rangle=\sum_{\mathbf{x} \in \mathbb{F}_p^n} P(f(\mathbf{x}))|\mathbf{x}\rangle$ can be expressed as
\begin{align}
|P(f)\rangle=\sum_{k=0}^{l} w_k\left|P^{(k)}\right\rangle,
\end{align}
where $w_0,...,w_l$ are coefficients that satisfy the normalization condition $\sum_k |w_k|^2 =1$.
We also denote  $\mathbf{w} = (w_0,...,w_l)$.

Substituting \eqref{37} into \eqref{47} yields
$$
\begin{aligned}
P^{(k)}\left(g_1\left(\mathbf{b}_1 \cdot \mathbf{x}\right), \ldots, g_m\left(\mathbf{b}_m \cdot \mathbf{x}\right)\right) & =\sum_{\substack{i_1, \ldots, i_k \\
\text { distinct }}} \prod_{\substack{i \in\left\{i_1, \ldots, i_k\right\} }} \left(\frac{1}{\sqrt{p}} \sum_{y_i \in \mathbb{F}_p} \omega_p^{-y_i \mathbf{b}_i \cdot \mathbf{x}} \tilde{g}_i\left(y_i\right)\right) \\
& =\sum_{\substack{\mathbf{y} \in \BFF_p^m \\
|\mathbf{y}| =k}} \frac{1}{\sqrt{p^k}} \omega_p^{-\left(B^T \mathbf{y}\right) \cdot \mathbf{x}} \prod_{\substack{i=1 \\
y_i \neq 0 }}^m \tilde{g}_i\left(y_i\right).
\end{aligned}
$$
Hence, the quantum Fourier transform of $\left|P^{(k)}\right\rangle$ is
\begin{align}\label{FT of P^(k)}
\left|\widetilde{P}^{(k)}\right\rangle:=F^{\otimes n}\left|P^{(k)}\right\rangle
=\frac{1}{\sqrt{\binom{m}{k}}} \sum_{\substack{\mathbf{y} \in \mathbb{F}_p^m \\|\mathbf{y}|=k}}\left(\prod_{\substack{i=1 \\ y_i \neq 0}}^m \tilde{g}_i\left(y_i\right)\right)\left|B^T \mathbf{y}\right\rangle,
\end{align}
where the transform $F$ has entries $F_{i j}=\omega_p^{i j} / \sqrt{p}$ with $i, j=0,..., p-1$.
If $|\mathbf{y}|<d^{\perp} / 2$, then $B^T \mathbf{y}$ are all distinct. 
Therefore, if $l<d^{\perp} / 2$ (where $d^{\perp}$ is the minimal distance of the code $\ker(B^T)$), then $\set{\left|\widetilde{P}^{(0)}\right\rangle, \ldots,\left|\widetilde{P}^{(l)}\right\rangle}$ form an orthonormal set and so do $\set{\left|P^{(0)}\right\rangle, \ldots,\left|P^{(l)}\right\rangle}$.
And the quantum Fourier transform of the DQI state $|P(f)\rangle$ is
\begin{align}
\left|\widetilde{P}(f)\right\rangle 
= & \sum_{k=0}^{l} w_k\left|\widetilde{P}^{(k)}\right\rangle \nonumber\\
= & \sum_{k=0}^{l} \frac{w_k}{\sqrt{\binom{m}{k}}} \sum_{\substack{\mathbf{y} \in \mathbb{F}_p^m \\|\mathbf{y}|=k}}\left(\prod_{\substack{i=1 \\ y_i \neq 0}}^m \tilde{g}_i\left(y_i\right)\right)\left|B^T \mathbf{y}\right\rangle \nonumber\\
=& \sum_{k=0}^{l} \frac{w_k}{\sqrt{\binom{m}{k}}} \sum_{\substack{\mathbf{y} \in \mathbb{F}_p^m \\|\mathbf{y}|=k}}
\tilde{g}(\mathbf{y})
\left|B^T \mathbf{y}\right\rangle, \label{66}
\end{align}
where we denote 
\begin{align}
\tilde{g}(\mathbf{y}) = \prod_{\substack{i=1 \\ y_i \neq 0}}^m \tilde{g}_i\left(y_i\right).
\end{align}
By convention, we take the empty product to be 1; in particular, this implies that $\tilde{g}(\mathbf{0})=1$. 

Without the condition $|\mathbf{y}|<d^{\perp} / 2$,
the states $\left|\widetilde{P}^{(0)}\right\rangle, \ldots,\left|\widetilde{P}^{(l)}\right\rangle$ are no longer orthogonal to each other. In this case, the DQI state $|P(f)\rangle=\sum_{\mathbf{x} \in \mathbb{F}_p^n} P(f(\mathbf{x}))|\mathbf{x}\rangle$ is still defined as the linear sum of states $|P^{(k)}\rangle$ with coefficients $w_0,...,w_l$,
whose quantum Fourier transforms $\left|\widetilde{P}^{(k)}\right\rangle$ satisfy \eqref{FT of P^(k)} as well.
From \eqref{66} we can get the following lemma (also see Lemma 10.1 in \cite{jordan2024optimization}).
\begin{lem}\label{norm of FT of P(f)}
The squared norm of $\left|\widetilde{P}(f)\right\rangle $ is
\begin{align}\label{144}
\ep{\widetilde{P}(f)| \widetilde{P}(f) } = \mathbf{w}^\dag M^{(m,l)} \mathbf{w},
\end{align}
where $M^{(m,l)}$ is the $(l+1)\times (l+1)$ symmetric matrix defined by
\begin{align}
M^{(m,l)}_{k,k'} = \frac{1}{\sqrt{ \binom{m}{k} \binom{m}{k'}  }} 
\sum_{\substack{\mathbf{y} \in \mathbb{F}_p^m \\|\mathbf{y}|=k}}
\sum_{\substack{\mathbf{y'} \in \mathbb{F}_p^m \\|\mathbf{y'}|=k'}}
\tilde{g}(\mathbf{y})^* \tilde{g}(\mathbf{y'})
\delta_{B^T \mathbf{y}, B^T \mathbf{y'} } .
\end{align}
\end{lem}

\section{Proof of Theorem~\ref{thm:main_1}}\label{appen:A}
To prove the theorem, we first need several technical lemmas.

\begin{lem}\label{250624lem1}
Let $p$ be a prime.
Let $a_1,a_2,...,a_l\in \BFF_p^*$ and $ \mathbf{a} = (a_1,a_2, ...,a_l)$.
For $0\le t \le l$, denote $P(t)$ to be the probability of $\ep{\mathbf{a}, \mathbf{v}} =0$ when $\mathbf{v}$ is chosen uniformly randomly from all vectors in $\BFF_p^{l}$ with Hamming weight $t$.
Then  
\begin{align}\label{250623eq2}
    P(t) = \frac 1p + \left(\frac{-1}{p-1} \right)^t \frac{p-1}{p}.
\end{align}
\end{lem}
\begin{proof}
We prove the statement by mathematical induction.
First, for $t=0$,  the equation \eqref{250623eq2} holds as $P(0)=1$.
Let us assume that \eqref{250623eq2} holds for $t$ and consider a vector $\mathbf{v}=(v_1,\ldots,v_l)$ with Hamming weight $t+1$.
Without loss of generality,  let us assume $v_1,...,v_{t+1} \neq 0$.
Then $\ep{\mathbf{a}, \mathbf{v}} =0$ if and only if $v_2a_2+\cdots + v_la_l \neq 0$ and $v_1a_1 = -(v_2a_2+\cdots + v_la_l)$.
The probability that $v_2a_2+\cdots + v_la_l \neq 0$ is $1-P(t)$,
and the probability that $v_1a_1 = -(v_2a_2+\cdots + v_la_l)$ when $v_2a_2+\cdots + v_la_l$ is given in $\BFF_p^*$ is $1/(p-1)$.
Hence, we have the following inductive relation,
\begin{align}
    P(t+1) = \frac{1-P(t)}{p-1},\quad \forall 0\le t\le l-1. 
\end{align}
This concludes the proof of \cref{250623eq2}.
\end{proof}

\begin{lem}\label{250627lem2}
Given a matrix $B$ with $\mathbf{b}_i$ being its $i$-th row, and  $Q(t,i)$ to be the probability of $\ep{\Buu, \mathbf{b}_i}=0$ for $\Buu $ being uniformly chosen from the set $\{\Buu \in \BFF_p^n: |\Buu| =t\}$, we have the following equality
\begin{align}
    \sum_{t=0}^n \binom{n}{t}(p-1)^t (\varepsilon/p)^{t} (1- (p-1)\varepsilon/p)^{n-t}\left(\frac{p}{p-1}  Q(t,i)-\frac{1}{p-1}\right)
    =(1-\varepsilon)^{|\mathbf{b}_i|},
\end{align} 
where $|\vec \alpha|$ and $|\mathbf{b}_i|$ denote the number of non-zero entries of the vectors $\vec \alpha$ and $\mathbf{b}_i$.
\end{lem}
\begin{proof}
First, we can simplify the equation as follows
    \begin{align*}
 & \sum_{t=0}^n \binom{n}{t}(p-1)^t (\varepsilon/p)^{t} (1- (p-1)\varepsilon/p)^{n-t}\left(\frac{p}{p-1}  Q(t,i)-\frac{1}{p-1}\right)\\
= &-\frac{1}{p-1}+ \frac{p}{p-1} \sum_{t=0}^n \binom{n}{t}(p-1)^t (\varepsilon/p)^{t} (1- (p-1)\varepsilon/p)^{n-t}Q(t,i).
\end{align*}

Let us denote $\vec \alpha=(\alpha_1,...,\alpha_n)$ to be the random vector in $\mathbb{F}^n_p$ 
with each $\alpha_i$ is chosen according to $\text{Pr}[\alpha_i=0]=(p-1)\epsilon/p$ and 
$\text{Pr}[\alpha_i=k]=1/(p-1)-\epsilon/p,\forall k\neq 0$. 
Hence,
\begin{align}\label{250627eq1}
  \sum_{t=0}^n \binom{n}{t}(p-1)^t (\varepsilon/p)^{t} (1- (p-1)\varepsilon/p)^{n-t}Q(t,i)    
\end{align}
is the probability that $ \ep{\vec\alpha,\mathbf{b}_i}=0 $.

Denote $P(s)$ to be the conditional probability of $\ep{\vec\alpha, \mathbf{b}_i} =0$ under the condition that 
$|\mathrm{supp}(\vec\alpha) \cap \mathrm{supp}(\mathbf{b}_i)| = s$. 
By Lemma \ref{250624lem1}, we have
\begin{align*} 
    P(s) = \frac 1p + \left(\frac{-1}{p-1} \right)^s \frac{p-1}{p}.
\end{align*}
Therefore,
\begin{align*}
&-\frac{1}{p-1}+ \frac{p}{p-1} \mathbb{P}\left( \ep{\vec\alpha,\mathbf{b}_i}=0 \right)\\
&= -\frac{1}{p-1}+ \frac{p}{p-1} \sum_{s=0}^l P(s) \mathbb{P}\left( |\mathrm{supp}(\vec\alpha) \cap \mathrm{supp}(\mathbf{b}_i)| = s \right) \\
&= -\frac{1}{p-1}+ \frac{p}{p-1} \sum_{s=0}^l 
\left[ \frac 1p + \left(\frac{-1}{p-1} \right)^s \frac{p-1}{p} \right]
\binom{l}{s} 
\left((p-1)\varepsilon/p\right)^s 
\left(1-(p-1)\varepsilon/p\right)^{l-s} \\
&= -\frac{1}{p-1}+ 
\frac{1}{p-1} \left[ (p-1)\varepsilon/p + 1-(p-1)\varepsilon/p \right]^l   +\sum_{s=0}^l (-1)^s  
\binom{l}{s} 
\left( \varepsilon/p\right)^s 
\left(1-(p-1)\varepsilon/p\right)^{l-s}\\
&= (1-\varepsilon)^l.
\end{align*}

\end{proof}

Now, we are ready to prove Theorem~\ref{thm:main_1}.

\begin{proof}[Proof of Theorem~\ref{thm:main_1}]
The proof is inspired by~\cite{jordan2024optimization}, and we focus on the effect of the 
noise here.
Let us define $s(\mathbf{x})$ to be the number of constraints 
satisfied by  $\mathbf{x} \in \mathbb{F}_p^n$ as follows
\begin{align}\label{71}
s(\mathbf{x})=\sum_{i=1}^m \mathbb{1}_{F_i}\left(\mathbf{b}_i \cdot \mathbf{x}\right),    
\end{align}
where $\mathbb{1}_{F_i}(x)$ denotes the indicator function for the set $F_i$:
$$
\mathbb{1}_{F_i}(x)= \begin{cases}1 & \text { if } x \in F_i; \\ 0 & \text { otherwise. }\end{cases}
$$

Since the indicator function 
can be written as $\mathbb{1}_{F_i}(x)=\sum_{v \in F_i} \mathbb{1}_{\{v\}}(x)=\frac{1}{p} \sum_{v \in F_i} \sum_{a \in \mathbb{F}_p} \omega_p^{a(x-v)}$,
the equation \eqref{71} can be written as
\begin{align*}
 s(\mathbf{x})=\frac{1}{p} \sum_{i=1}^m \sum_{v \in F_i} \sum_{a \in \mathbb{F}_p} \omega_p^{a\left(\mathbf{b}_i \cdot \mathbf{x}-v\right)}.
\end{align*}

The expected number of constraints satisfied by a symbol string sampled from the output distribution of the errored DQI state
$\CEE^{\otimes n} (\ket{P(f)}\bra{P(f)})$ is given by
\begin{align*}
    \left\langle s_D^{(m, l)}\right\rangle=\trace\left[  S_f \CEE^{\otimes n} (\ket{P(f)}\bra{P(f)}) \right] = \trace\left[ \CEE^{\otimes n}( S_f)  \ket{P(f)}\bra{P(f)} \right],
\end{align*}
where
\begin{align*}
    S_f=\sum_{\mathbf{x} \in \mathbb{F}_p^n} s(\mathbf{x})|\mathbf{x}\rangle\langle\mathbf{x}|.  
\end{align*}
We can rewrite  $S_f$ in terms of the Pauli operator $Z=\sum_{x \in \mathbb{F}_p} \omega_p^x|x\rangle\langle x|$ as
\begin{align}
\nonumber S_f & =\sum_{\mathbf{x} \in \mathbb{F}_p^n} s(\mathbf{x})|\mathbf{x}\rangle\langle\mathbf{x}|  =\frac{1}{p} \sum_{i=1}^m \sum_{v \in F_i} \sum_{a \in \mathbb{F}_p} \sum_{\mathbf{x} \in \mathbb{F}_p^n} \omega_p^{a\left(\mathbf{b}_i\cdot \mathbf{x}-v\right)}|\mathbf{x}\rangle\langle\mathbf{x}| \\
\nonumber & =\frac{1}{p} \sum_{i=1}^m \sum_{v \in F_i} \sum_{a \in \mathbb{F}_p} \omega_p^{-a v} \bigotimes_{j=1}^n \sum_{x_j \in \mathbb{F}_p} \omega_p^{a B_{i j} x_j}\left|x_j\right\rangle\left\langle x_j\right| \\
& =\frac{1}{p} \sum_{i=1}^m \sum_{v \in F_i} \sum_{a \in \mathbb{F}_p} \omega_p^{-a v} \prod_{j=1}^n Z_j^{a B_{i j}} . \label{80}
\end{align}

Therefore 
\begin{align*}
 (X^{\vec \alpha}Z^{\vec \beta}) S_f (X^{\vec \alpha}Z^{\vec \beta})^\dag
=&\frac{1}{p} \sum_{i=1}^m \sum_{v \in F_i} \sum_{a \in \mathbb{F}_p} \omega_p^{-a v}   (X^{\vec \alpha}Z^{\vec \beta}) Z^{a \mathbf{b}_i} (X^{\vec \alpha}Z^{\vec \beta})^\dag\\
=& \frac{1}{p} \sum_{i=1}^m \sum_{v \in F_i} \sum_{a \in \mathbb{F}_p} \omega_p^{-a v+a \ep{\mathbf{b}_i, \vec \alpha}}    Z^{a \mathbf{b}_i},  
\end{align*}
where we used the fact that $X^{\vec \alpha} Z^{\vec \beta} = \omega_p^{\ep{\vec \alpha, \vec \beta}}Z^{\vec \beta}  X^{\vec \alpha}$.
Hence, the action of  $\CEE^{\otimes n}$ acting on $ S_f$ is equivalent to the channel $\CEE_1^{\otimes n}$ acting on $ S_f$, where 
\[\CEE_1(\rho) = (1-\varepsilon)\rho + \frac{\varepsilon}{p}\sum_{\alpha\in \BFF_p} X^\alpha \rho X^{-\alpha} , \]
and
\begin{align*}
\left\langle s_D^{(m, l)}\right\rangle 
=&\trace\left[  \CEE_1^{\otimes n} (S_f ) \ket{P(f)}\bra{P(f)} \right]\\
=& \sum_{\vec\alpha\in \BFF_p^n} (\varepsilon/p)^{|\vec\alpha|} (1-(p-1)\varepsilon/p)^{1- |\vec\alpha|}   \bra{P(f)} X^{\vec \alpha} S_f X^{-\vec \alpha} \ket{P(f)} .  
\end{align*}
For each $\vec\alpha\in \BFF_p^n$, using equation \eqref{80}, we obtain
\begin{align}
\nonumber &\bra{P(f)} X^{\vec \alpha} S_f X^{-\vec \alpha} \ket{P(f)}\\
\nonumber=& \frac{1}{p} \sum_{i=1}^m \sum_{v \in F_i} \sum_{a \in \mathbb{F}_p} \omega_p^{-a v}\langle {P}(f)| X^{\vec \alpha}  Z_j^{ a \mathbf{b}_i} X^{-\vec \alpha}    | {P}(f)\rangle\\
\nonumber=& \frac{1}{p} \sum_{i=1}^m \sum_{v \in F_i} \sum_{a \in \mathbb{F}_p} \omega_p^{-a v+a \ep{\mathbf{b}_i,\vec \alpha }}\langle {P}(f)|  Z_j^{ a \mathbf{b}_i}     | {P}(f)\rangle\\
=& \frac{1}{p} \sum_{i=1}^m \sum_{v \in F_i} \sum_{a \in \mathbb{F}_p} \omega_p^{-a v+a \ep{\mathbf{b}_i,\vec \alpha }}\langle\widetilde{P}(f)|  X ^{-a \mathbf{b}_i}|\widetilde{P}(f)\rangle,
\label{82}
\end{align}
where  $F Z F^{\dagger}=X^{-1}$ and $|\widetilde{P}(f)\rangle=F^{\otimes n}|P(f)\rangle$. 
By substituting \cref{66} into \cref{82}, we obtain
\begin{align*}
&\bra{P(f)} X^{\vec \alpha} S_f X^{-\vec \alpha} \ket{P(f)}\\
=&  \frac{1}{p} \sum_{k_1, k_2=0}^{l} \frac{w_{k_1}^* w_{k_2}}{\sqrt{\binom{m}{k_1}\binom{m}{k_2}}} 
\sum_{\substack{\mathbf{y}_1, \mathbf{y}_2 \in \mathbb{F}_p^m \\\left|\mathbf{y}_1\right|=k_1 \\\left|\mathbf{y}_2\right|=k_2}} \tilde{g}^*\left(\mathbf{y}_1\right) \tilde{g}\left(\mathbf{y}_2\right) 
\sum_{i=1}^m \sum_{v \in F_i} \sum_{\alpha \in \mathbb{F}_p}
\omega_p^{-a v+a \ep{\mathbf{b}_i, \vec \alpha}}
\left\langle B^T \mathbf{y}_1\right|  X ^{-a \mathbf{b}_i} \left|B^T \mathbf{y}_2\right\rangle .
\end{align*}

Let $\mathbf{e}_1, \ldots, \mathbf{e}_m \in \mathbb{F}_p^m$ denote the standard basis of one-hot vectors. Then
\begin{align}\label{85}
\begin{aligned}
&\bra{P(f)} X^{\vec \alpha} S_f X^{-\vec \alpha} \ket{P(f)}\\
=&\frac{1}{p} \sum_{k_1, k_2=0}^{l} \frac{w_{k_1}^* w_{k_2}}{\sqrt{\binom{m}{k_1}\binom{m}{k_2}}} \sum_{\substack{\mathbf{y}_1, \mathbf{y}_2 \in \mathbb{F}_p^m \\\left|\mathbf{y}_1\right|=k_1 \\\left|\mathbf{y}_2\right|=k_2}} 
\tilde{g}^*\left(\mathbf{y}_1\right) \tilde{g} \left(\mathbf{y}_2\right) 
\sum_{i=1}^m \sum_{v \in F_i} \sum_{a \in \mathbb{F}_p} 
\omega_p^{-a v+a \ep{\mathbf{b}_i, \vec \alpha}}
\left\langle B^T \mathbf{y}_1 \bigg| B^T\left(\mathbf{y}_2-a \mathbf{e}_i\right)\right\rangle .
\end{aligned}
\end{align}

Since both states $\left|B^T \mathbf{y}_1\right\rangle$ and $\left|B^T\left(\mathbf{y}_2-a \mathbf{e}_i\right)\right\rangle$ are computational-basis states, we have
$$
\left\langle B^T \mathbf{y}_1 \bigg| B^T\left(\mathbf{y}_2-a \mathbf{e}_i\right)\right\rangle= \begin{cases}1 & \text { if } B^T \mathbf{y}_1=B^T\left(\mathbf{y}_2-a \mathbf{e}_i\right), \\ 0 & \text { otherwise. }\end{cases}
$$
Moreover,
$$
B^T \mathbf{y}_1=B^T\left(\mathbf{y}_2-a \mathbf{e}_i\right) \Longleftrightarrow \mathbf{y}_1-\mathbf{y}_2+a \mathbf{e}_i \in \ker{B^T} \Longleftrightarrow \mathbf{y}_1=\mathbf{y}_2-a \mathbf{e}_i,
$$
where we used the assumption that the smallest Hamming weight of a non-zero symbol string in $\ker{B^T}$ is $d^{\perp}>2 l+1 \geq k_1+k_2+1$.
Hence, there are four possible cases in which $\left\langle B^T \mathbf{y}_1 \bigg| B^T\left(\mathbf{y}_2-a \mathbf{e}_i\right)\right\rangle $ can be non-zero: 

(I). $\left|\mathbf{y}_1\right|=\left|\mathbf{y}_2\right|-1$, 

(II). $\left|\mathbf{y}_2\right|=\left|\mathbf{y}_1\right|-1$, 

(III). $\left|\mathbf{y}_1\right|=\left|\mathbf{y}_2\right|$ and $\mathbf{y}_1 \neq \mathbf{y}_2$, 

(IV). $\mathbf{y}_1 = \mathbf{y}_2$.\\
Therefore, the equation \eqref{85} can be split into 4 parts:
\begin{align*}
&\bra{P(f)} X^{\vec \alpha} S_f X^{-\vec \alpha} \ket{P(f)} \\
= & \frac{1}{p} \sum_{k=0}^{l-1} \frac{w_k^* w_{k+1}}{\sqrt{\binom{m}{k}\binom{m}{k+1}}} \sum_{\substack{\mathbf{y} \in \mathbb{F}_p^m \\
|\mathbf{y}|=k}}|\tilde{g}(\mathbf{y})|^2 \sum_{\substack{i=1 \\
y_i=0}}^m \sum_{v \in F_i} \sum_{a \in \mathbb{F}_p^*} 
\omega_p^{-a v+a \ep{\mathbf{b}_i, \vec \alpha}}
\tilde{g}_i(a) \\
& + \frac{1}{p} \sum_{k=0}^{l-1} \frac{w_{k+1}^* w_k}{\sqrt{\binom{m}{k+1}\binom{m}{k}}} \sum_{\substack{\mathbf{y} \in \mathbb{F}_p^m \\
|\mathbf{y}|=k}}|\tilde{g}(\mathbf{y})|^2 \sum_{\substack{i=1 \\
y_i=0}}^m \sum_{v \in F_i} \sum_{a \in \mathbb{F}_p^*} 
\omega_p^{-a v+a \ep{\mathbf{b}_i, \vec \alpha}}
\tilde{g}_i(a)\\
& + \frac{1}{p} \sum_{k=1}^{l} \frac{\left|w_k\right|^2}{\binom{m}{k}} \sum_{\substack{\mathbf{y} \in \mathbb{F}_p^m \\
|\mathbf{y}|=k-1}}|\tilde{g}(\mathbf{y})|^2 \sum_{\substack{i=1 \\
y_i=0}}^m \sum_{v \in F_i} \sum_{a \in \mathbb{F}_p^*} \sum_{z \in \mathbb{F}_p \backslash\{0, a\}} 
\omega_p^{-a v+a \ep{\mathbf{b}_i, \vec \alpha}}
\tilde{g}_i(a-z) \tilde{g}_i(z)\\
& + \frac{1}{p} \sum_{k=0}^{l} \frac{\left|w_k\right|^2}{\binom{m}{k}} \sum_{\substack{\mathbf{y} \in \mathbb{F}_p^m \\
|\mathbf{y}|=k}}|\tilde{g}(\mathbf{y})|^2 \sum_{i=1}^m \sum_{v \in F_i} \sum_{a \in\{0\}} 
\omega_p^{-a v+a \ep{\mathbf{b}_i, \vec \alpha}},
\end{align*}
and correspondingly, $ \left\langle s_D^{(m, l)}\right\rangle $ can also be split into 4 parts
\begin{align*}
 \left\langle s_D^{(m, l)}\right\rangle 
=& \sum_{\vec\alpha\in \BFF_p^n} (\varepsilon/p)^{|\vec\alpha|} (1-(p-1)\varepsilon/p)^{1- |\vec\alpha|}   \bra{P(f)} X^{\vec \alpha} S_f X^{-\vec \alpha} \ket{P(f)} \\
=& (I)+(II)+(III)+(IV).
\end{align*}

(I). First, we have
$$
 \sum_{a \in \mathbb{F}_p^*} 
\omega_p^{-a v+a \ep{\mathbf{b}_i, \vec \alpha} }
\tilde{g}_i(a)
=\sqrt{p}  g_i(v-\ep{\mathbf{b}_i, \vec \alpha}) ,
$$
where we used the fact that $\tilde{g}_i(0) =0 $.
Hence, the first part
\begin{align*}
(I) = & \sum_{\vec\alpha\in \BFF_p^n} (\varepsilon/p)^{|\vec\alpha|} (1-(p-1)\varepsilon/p)^{1- |\vec\alpha|}  \frac{1}{p} \sum_{k=0}^{l-1} \frac{w_k^* w_{k+1}}{\sqrt{\binom{m}{k}\binom{m}{k+1}}} \sum_{\substack{\mathbf{y} \in \mathbb{F}_p^m \\
|\mathbf{y}|=k}}|\tilde{g}(\mathbf{y})|^2 \\
&\times\sum_{\substack{i=1 \\
y_i=0}}^m \sum_{v \in F_i} \sum_{a \in \mathbb{F}_p^*} 
\omega_p^{-a v+a \ep{\mathbf{b}_i, \vec \alpha}}
\tilde{g}_i(a)\\
=& \sum_{\vec\alpha\in \BFF_p^n} (\varepsilon/p)^{|\vec\alpha|} (1-(p-1)\varepsilon/p)^{1- |\vec\alpha|}  \frac{1}{p} \sum_{k=0}^{l-1} \frac{w_k^* w_{k+1}}{\sqrt{\binom{m}{k}\binom{m}{k+1}}} \sum_{\substack{\mathbf{y} \in \mathbb{F}_p^m \\
|\mathbf{y}|=k}}|\tilde{g}(\mathbf{y})|^2 \\
&\times\sum_{\substack{i=1 \\
y_i=0}}^m \sum_{v \in F_i} \sqrt{p}  g_i(v- \ep{\mathbf{b}_i, \vec \alpha}) \\
=& \sum_{t=0}^n  (\varepsilon/p)^{t} (1-(p-1)\varepsilon/p)^{1- t}  \frac{1}{p} \sum_{k=0}^{l-1} \frac{w_k^* w_{k+1}}{\sqrt{\binom{m}{k}\binom{m}{k+1}}} \sum_{\substack{\mathbf{y} \in \mathbb{F}_p^m \\
|\mathbf{y}|=k}}|\tilde{g}(\mathbf{y})|^2\\
&\times\sum_{\substack{i=1 \\
y_i=0}}^m \sum_{v \in F_i} \sqrt{p}  
\sum_{\substack{\vec\alpha\in \BFF_p^n \\ |\vec\alpha | =t}}
g_i(v- \ep{\mathbf{b}_i, \vec \alpha}).
\end{align*}
Let us denote $Q(t,i)$ to be the probability of $\ep{\Buu, \mathbf{b}_i}=0$ when $\Buu $ is uniformly chosen from the set $\{\Buu \in \BFF_p^n: |\Buu| =t\}$.
Then, $\sum_{\substack{\vec\alpha\in \BFF_p^n \\ |\vec\alpha | =t}}
g_i(v- \ep{\mathbf{b}_i, \vec \alpha})$ can be written as 
\begin{align*}
    \sum_{\substack{\vec\alpha\in \BFF_p^n \\ |\vec\alpha | =t}}
g_i(v- \ep{\mathbf{b}_i, \vec \alpha})
=&\sum_{\substack{\vec\alpha\in \BFF_p^n \\ |\vec\alpha | =t}}
\left[Q(t,i) g_i(v) +(1- Q(t,i)) \Expect_{z\in \BFF_p\setminus \{v\} }g_i(z)
\right]\\
=&\sum_{\substack{\vec\alpha\in \BFF_p^n \\ |\vec\alpha | =t}}
\left(\frac{p}{p-1} Q(t,i)-\frac{1}{p-1}\right) g_i(v),
\end{align*}
where we have used the fact that $\sum_z g_i(z) =0$.
Hence, we have 
\begin{align*}
(I) = & \sum_{t=0}^n  (\varepsilon/p)^{t} (1-(p-1)\varepsilon/p)^{1- t}  \frac{1}{p} \sum_{k=0}^{l-1} \frac{w_k^* w_{k+1}}{\sqrt{\binom{m}{k}\binom{m}{k+1}}} \sum_{\substack{\mathbf{y} \in \mathbb{F}_p^m \\
|\mathbf{y}|=k}}|\tilde{g}(\mathbf{y})|^2  \\
&\times\sum_{\substack{i=1 \\
y_i=0}}^m\sum_{v \in F_i} \sqrt{p} 
\sum_{\substack{\vec\alpha\in \BFF_p^n \\ |\vec\alpha | =t}}
\left(\frac{p}{p-1} Q(t,i)-\frac{1}{p-1}\right) g_i(v).
\end{align*}


Due to the fact $g_i(v) = \frac{1- \bar f_i}{\varphi} = \sqrt{\frac{p-r}{pr}}
$ for $v\in F_i$, we have
\begin{align*}
(I) &=  \frac{\sqrt{(p-r)r}}{p} \sum_{\vec\alpha\in \BFF_p^n} (\varepsilon/p)^{|\vec\alpha|} (1-(p-1)\varepsilon/p)^{1- |\vec\alpha|}
\sum_{k=0}^{l-1} \frac{w_k^* w_{k+1}}{\sqrt{\binom{m}{k}\binom{m}{k+1}}} \sum_{\substack{\mathbf{y} \in \mathbb{F}_p^m \\ |\mathbf{y}|=k}}
|\tilde{g}(\mathbf{y})|^2 \\
&\quad\times
\sum_{\substack{i=1 \\ y_i=0}}^m 
\left(\frac{p}{p-1} Q(|\vec\alpha| ,i)   -\frac{1}{p-1}\right) \\
&= \frac{\sqrt{(p-r)r}}{p} 
\sum_{\vec\alpha\in \BFF_p^n}
(\varepsilon/p)^{|\vec\alpha|} (1-(p-1)\varepsilon/p)^{1- |\vec\alpha|}
\sum_{k=0}^{l-1} 
\frac{w_k^* w_{k+1}}{\sqrt{\binom{m}{k}\binom{m}{k+1}}} \\
&\quad \times
\sum_{\substack{i_1, \ldots, i_k \in [m] \\ \text { distinct }}}
\left(
\sum_{\substack{i\in [m]\setminus \{i_1, \ldots, i_k\}}}  
\left(\frac{p}{p-1} Q(|\vec\alpha| ,i)   -\frac{1}{p-1}\right)
\right)
\sum_{y_1,...,y_k\in \BFF_p^*}|\widetilde{g}_{i_1}(y_1) \cdots \widetilde{g}_{i_k}(y_k)|^2 .
\end{align*}
Since $\tilde{g}_i(0)=0$, and $\sum_{y\in \BFF_p} |\tilde{g}_i(y)|^2=1$ for all $i$, we have
\[\sum_{y_1,...,y_k\in \BFF_p^*}|\widetilde{g}_{i_1}(y_1) \cdots \widetilde{g}_{i_k}(y_k)|^2 =1.\]
Moreover, since
\begin{align*}
&\sum_{\substack{i_1, \ldots, i_k \in [m] \\ \text { distinct }}}
\left(
\sum_{\substack{i\in [m]\setminus \{i_1, \ldots, i_k\}}}  
\left(\frac{p}{p-1} Q(|\vec\alpha| ,i)   -\frac{1}{p-1}\right)
\right) \\
&= (m-k) \binom{m}{k} \Expect_{i\in [m]}  
\left(\frac{p}{p-1} Q(|\vec\alpha| ,i)   -\frac{1}{p-1}\right),
\end{align*}
then we have
\begin{align*} 
(I) &= \frac{\sqrt{(p-r)r}}{p} 
\sum_{k=0}^{l-1} 
\frac{w_k^* w_{k+1}}{\sqrt{\binom{m}{k}\binom{m}{k+1}}} (m-k) \binom{m}{k} \\
&\quad\times
\sum_{\vec\alpha\in \BFF_p^n}
(\varepsilon/p)^{|\vec\alpha|} (1-(p-1)\varepsilon/p)^{1- |\vec\alpha|}
\Expect_{i\in [m]}  
\left(\frac{p}{p-1} Q(|\vec\alpha| ,i)   -\frac{1}{p-1}\right) \\
&= \frac{\sqrt{(p-r)r}}{p} \sum_{k=0}^{l-1} 
 w_k^* w_{k+1}  \sqrt{(k+1)(m-k)}  \\
 &\quad\times
\Expect_{i\in [m]}  
 \sum_{\vec\alpha\in \BFF_p^n}
(\varepsilon/p)^{|\vec\alpha|} (1-(p-1)\varepsilon/p)^{1- |\vec\alpha|}
\left(\frac{p}{p-1} Q(|\vec\alpha| ,i)   -\frac{1}{p-1}\right)\\
&= \frac{\sqrt{(p-r)r}}{p} \sum_{k=0}^{l-1} 
 w_k^* w_{k+1}  \sqrt{(k+1)(m-k)}  \\
&\quad\times
\Expect_{i\in [m]}  
   \sum_{t=0}^n \binom{n}{t}(p-1)^t (\varepsilon/p)^{t} (1- (p-1)\varepsilon/p)^{n-t}\left(\frac{p}{p-1}  Q(t,i)-\frac{1}{p-1}\right)\\
&=\frac{\sqrt{(p-r)r}}{p} \sum_{k=0}^{l-1} 
 w_k^* w_{k+1}  \sqrt{(k+1)(m-k)}  
\Expect_{i\in [m]}  
(1-\varepsilon)^{|\mathbf{b}_i|}\\
&=\frac{\sqrt{(p-r)r}}{p} \sum_{k=0}^{l-1} 
 w_k^* w_{k+1}  \sqrt{(k+1)(m-k)}  
\tau_1(B,\varepsilon),
\end{align*}
where the second to the last line comes from Lemma~\ref{250627lem2}.

(II). This case is similar to (I), we have
\begin{align*} 
(II) = \frac{\sqrt{(p-r)r}}{p} \tau_1(B,\varepsilon) \sum_{k=0}^{l-1} 
 w_k w_{k+1}^*  \sqrt{(k+1)(m-k)}.
\end{align*}

(III).  In this case, we have

\begin{align*}
    &\sum_{a \in \mathbb{F}_p} \omega_p^{-a v+a \ep{\mathbf{b}_i, \vec \alpha}} 
\sum_{z \in \mathbb{F}_p} \tilde{g}_i(a-z) \tilde{g}_i(z) \\
=&\frac{1}{p} \sum_{a \in \mathbb{F}_p} 
\omega_p^{a (-v+ a \ep{\mathbf{b}_i, \vec \alpha})} \sum_{z \in \mathbb{F}_p} \sum_{x \in \mathbb{F}_p} \omega_p^{x(a-z)} g_i(x) \sum_{y \in \mathbb{F}_p} \omega_p^{y z} g_i(y) \\
=&\sum_{a, x, y \in \mathbb{F}_p} \omega_p^{a(x-v+ \ep{\mathbf{b}_i, \vec \alpha})} g_i(x) g_i(y) \frac{1}{p} \sum_{z \in \mathbb{F}_p} \omega_p^{(y-x) z} \\
=&\sum_{x \in \mathbb{F}_p} g_i(x)^2 \sum_{a \in \mathbb{F}_p} \omega_p^{a(x-v+ \ep{\mathbf{b}_i, \vec \alpha})}\\
=& p g_i(v-\ep{\mathbf{b}_i, \vec \alpha})^2.
\end{align*}

Hence,
\begin{align*}
(III) = & 
\sum_{\vec\alpha\in \BFF_p^n} (\varepsilon/p)^{|\vec\alpha|} (1-(p-1)\varepsilon/p)^{1- |\vec\alpha|}  
\frac{1}{p} \sum_{k=1}^{l} \frac{\left|w_k\right|^2}{\binom{m}{k}} \sum_{\substack{\mathbf{y} \in \mathbb{F}_p^m \\
|\mathbf{y}|=k-1}} |\tilde{g}(\mathbf{y})|^2 \\
& \quad \times
\sum_{\substack{i=1 \\y_i=0}}^m 
\sum_{v \in F_i} \sum_{a \in \mathbb{F}_p^*} 
\sum_{z \in \mathbb{F}_p \backslash\{0, a\}} 
\omega_p^{-a v+a \ep{\mathbf{b}_i, \vec \alpha}}
\tilde{g}_i(a-z) \tilde{g}_i(z)\\
=& \sum_{\vec\alpha\in \BFF_p^n} (\varepsilon/p)^{|\vec\alpha|} (1-(p-1)\varepsilon/p)^{1- |\vec\alpha|}  
\frac{1}{p} \sum_{k=1}^{l} \frac{\left|w_k\right|^2}{\binom{m}{k}} \sum_{\substack{\mathbf{y} \in \mathbb{F}_p^m \\
|\mathbf{y}|=k-1}} |\tilde{g}(\mathbf{y})|^2 \\
& \quad \times
\sum_{\substack{i=1 \\y_i=0}}^m 
\sum_{v \in F_i} 
\left(
p g_i(v-\ep{\mathbf{b}_i, \vec \alpha})^2 -1
\right)\\
=& \sum_{t=0}^n (\varepsilon/p)^{t} (1-(p-1)\varepsilon/p)^{1- t}  
\frac{1}{p} \sum_{k=1}^{l} \frac{\left|w_k\right|^2}{\binom{m}{k}} \sum_{\substack{\mathbf{y} \in \mathbb{F}_p^m \\
|\mathbf{y}|=k-1}} |\tilde{g}(\mathbf{y})|^2 \\
& \quad \times
\sum_{\substack{i=1 \\y_i=0}}^m 
\sum_{v \in F_i} 
\sum_{\substack{ \vec\alpha\in \BFF_p^n \\ |\vec\alpha| =t}}
\left(
p g_i(v-\ep{\mathbf{b}_i, \vec \alpha})^2 -1
\right)\\
=&  \sum_{t=0}^n (\varepsilon/p)^{t} (1-(p-1)\varepsilon/p)^{1- t}  
\frac{1}{p} \sum_{k=1}^{l} \frac{\left|w_k\right|^2}{\binom{m}{k}} \sum_{\substack{\mathbf{y} \in \mathbb{F}_p^m \\
|\mathbf{y}|=k-1}} |\tilde{g}(\mathbf{y})|^2 \\
& \quad \times
\sum_{\substack{i=1 \\y_i=0}}^m 
\sum_{v \in F_i} 
\sum_{\substack{ \vec\alpha\in \BFF_p^n \\ |\vec\alpha| =t}}
\left(
Q(t,i) (p g_i(v)^2 -1) +(1-Q(t,i)) 
\Expect_{z\in \BFF_p\setminus \{v\}}(p g_i(z)^2 -1 )
\right).
\end{align*}
Due to the fact that  $\Expect_{z\in \BFF_p }(p g_i(z)^2 -1 )=0$, 
we have
\begin{align*}
(III)
=& \sum_{\substack{ \vec\alpha\in \BFF_p^n }} 
(\varepsilon/p)^{|\vec\alpha |} (1-(p-1)\varepsilon/p)^{1- |\vec\alpha |}  
\frac{1}{p} \sum_{k=1}^{l} \frac{\left|w_k\right|^2}{\binom{m}{k}} \sum_{\substack{\mathbf{y} \in \mathbb{F}_p^m \\
|\mathbf{y}|=k-1}} |\tilde{g}(\mathbf{y})|^2 \\
& \quad \times
\sum_{\substack{i=1 \\y_i=0}}^m 
\sum_{v \in F_i} 
\left(\frac{p}{p-1} Q(|\vec\alpha| ,i)   -\frac{1}{p-1}\right) (p g_i(v)^2 -1 )\\
=& \sum_{\substack{ \vec\alpha\in \BFF_p^n }} 
(\varepsilon/p)^{|\vec\alpha |} (1-(p-1)\varepsilon/p)^{1- |\vec\alpha |}  
\frac{1}{p} \sum_{k=1}^{l} \frac{\left|w_k\right|^2}{\binom{m}{k}} \sum_{\substack{\mathbf{y} \in \mathbb{F}_p^m \\
|\mathbf{y}|=k-1}} |\tilde{g}(\mathbf{y})|^2 \\
& \quad \times
\sum_{\substack{i=1 \\y_i=0}}^m  
\left(\frac{p}{p-1} Q(|\vec\alpha| ,i)   -\frac{1}{p-1}\right) (p-2r),
\end{align*}
where we used the fact that $g_i(v) = \frac{1- \bar f_i}{\varphi} = \sqrt{\frac{p-r}{pr}}
$ for $v\in F_i$.
In addition, due to the fact
$\sum_{y_1,...,y_k\in \BFF_p^*}|\widetilde{g}_{i_1}(y_1) \cdots \widetilde{g}_{i_k}(y_k)|^2 =1$, 
we get 
\begin{align*}
(III) = &
\frac{p-2r}{p}
\sum_{k=1}^{l} \frac{\left|w_k\right|^2}{\binom{m}{k}} 
\sum_{\substack{ \vec\alpha\in \BFF_p^n }} 
(\varepsilon/p)^{|\vec\alpha |} (1-(p-1)\varepsilon/p)^{1- |\vec\alpha |}  
\\
& \times
\sum_{\substack{i_1, \ldots, i_{k-1} \in [m] \\ \text { distinct }}} 
\left(\sum_{\substack{i\in [m]\setminus \{i_1, \ldots, i_{k-1}\}}}  
\left(\frac{p}{p-1} Q(|\vec\alpha| ,i)   -\frac{1}{p-1}\right)
\right)\\
&\times
\sum_{y_1,...,y_{k-1}\in \BFF_p^*}|\widetilde{g}_{i_1}(y_1) \cdots \widetilde{g}_{i_{k-1}}(y_{k-1})|^2\\
= & \frac{p-2r}{p}
\sum_{k=1}^{l}  \left|w_k\right|^2 k 
\sum_{\substack{ \vec\alpha\in \BFF_p^n }} 
(\varepsilon/p)^{|\vec\alpha |} (1-(p-1)\varepsilon/p)^{1- |\vec\alpha |}  
\Expect_{\substack{i\in [m] }}  
\left(\frac{p}{p-1} Q(|\vec\alpha| ,i)   -\frac{1}{p-1}\right)
 \\
=& \frac{p-2r}{p} \tau_1(B,\varepsilon)
\sum_{k=1}^{l}  \left|w_k\right|^2 k,
\end{align*}
where the last line comes from  Lemma~\ref{250627lem2}.

(IV). In this case, we have
\begin{align*}
(IV) = & 
\sum_{\vec\alpha\in \BFF_p^n} (\varepsilon/p)^{|\vec\alpha|} (1-(p-1)\varepsilon/p)^{1- |\vec\alpha|}  
\frac{1}{p} \sum_{k=1}^{l} \frac{\left|w_k\right|^2}{\binom{m}{k}} \sum_{\substack{\mathbf{y} \in \mathbb{F}_p^m \\
|\mathbf{y}|=k}} |\tilde{g}(\mathbf{y})|^2 
\sum_{\substack{i=1}}^m 
\sum_{v \in F_i} \sum_{a \in \{0\}}  
\omega_p^{-a v+a \ep{\mathbf{b}_i, \vec \alpha}} \\
=& \frac{mr}{p}.
\end{align*}

Finally, let us 
put all the things together, then we have
\begin{align*}
\left\langle s_D^{(m, l)}\right\rangle =& (I)+ (II)+(III)+(IV)\\
=& \frac{\sqrt{(p-r)r}}{p} \tau_1(B,\varepsilon) \sum_{k=0}^{l-1} 
 w_k^* w_{k+1}  \sqrt{(k+1)(m-k)} \\
 & + \frac{\sqrt{(p-r)r}}{p} \tau_1(B,\varepsilon) \sum_{k=0}^{l-1} 
 w_k w_{k+1}^*  \sqrt{(k+1)(m-k)}\\
 & + \frac{p-2r}{p} \tau_1(B,\varepsilon)
\sum_{k=1}^{l}  \left|w_k\right|^2 k +\frac{mr}{p}\\ 
=& \frac{mr}{p} +  \frac{\sqrt{r(p-r)}}{p} \tau_1(B,\varepsilon)  \mathbf{w}^{\dagger} A^{(m, l, d)} \mathbf{w} ,
\end{align*}
where $\mathbf{w}=\left(w_0, \ldots, w_{l}\right)^T$ and $A^{(m, l, d)}$ is defined in \eqref{111}.
\end{proof}

\section{Proof of Theorem~\ref{thm:main_3}}\label{appen:B}
To prove Theorem~\ref{thm:main_3}, we first need to prove several lemmas for the general setting.

\begin{lem}\label{thm:main_2}
Let $f(\mathbf{x})=\sum_{i=1}^m f_i\left(\sum_{j=1}^n B_{i j} x_j\right)$ be a \maxlinsat{} objective function with matrix $B \in \BFF_p^{m \times n}$ for a prime $p$ and positive integers $m$ and $n$ such that $m>n$. 
Suppose that $\left|f_i^{-1}(+1)\right|=r$ for some $r \in\{1, \ldots, p-1\}$. 
Let 
$P(f) $ be the degree-$l$ polynomial determined by coefficients $w_0,...,w_l$ such that the perfect DQI state $\ket{P(f)}$ satisfies \eqref{51} (note that $\ket{P(f)}$ is not normalized).
Let $\left\langle s_D^{(m, l)}\right\rangle$ be the expected number of satisfied constraints for the symbol string obtained upon measuring the errored imperfect DQI state $\CEE^{\otimes n} (\ket{P_\CDD(f)}\bra{P_\CDD(f)})$ in the computational basis. 
Then
$$
\left\langle s_D^{(m, l)}\right\rangle=  \frac{\mathbf{w}^{\dagger} \bar{A}^{(m, l,\CDD)} \mathbf{w}}{\ep{P_\CDD(f) \big| P_\CDD(f)}} ,
$$
where $\bar{A}^{(m, l,\CDD)}$ is the $(l+1) \times(l+1)$ symmetric matrix defined by
\begin{align}
\begin{aligned}\label{177}
\bar{A}^{(m, l,\CDD)}_{k_1,k_2} = &
  \frac{1}{\sqrt{\binom{m}{k_1}\binom{m}{k_2}}}
\sum_{\vec\alpha\in \BFF_p^n} (\varepsilon/p)^{|\vec\alpha|} (1-(p-1)\varepsilon/p)^{1- |\vec\alpha|}   
  \\
&\times
\frac{1}{p} 
\sum_{i=1}^m 
\sum_{v \in F_i} 
\sum_{a \in \mathbb{F}_p}
\sum_{\substack{(\mathbf{y}_1, \mathbf{y}_2) \in S^{(i,a,\CDD)}_{k_1,k_2}
}} 
\tilde{g}^*\left(\mathbf{y}_1\right) \tilde{g} \left(\mathbf{y}_2\right)
\omega_p^{-a v+a \ep{\mathbf{b}_i, \vec \alpha}},
\end{aligned}
\end{align} 
for $0\le k_1,k_2\le l$, and 
\begin{align}\label{def:setS}
S^{(i,a,\CDD)}_{k_1,k_2} = \left\{(\mathbf{y}_1, \mathbf{y}_2) \in \CDD_{k_1} \times \CDD_{k_2} : 
B^T(\mathbf{y}_1 - \mathbf{y}_2 + a\mathbf{e}_i)=\mathbf{0}\right \}.
\end{align}
\end{lem}
\begin{proof}
Similar to the proof of Theorem \ref{thm:main_1}, we have
\begin{align*}
    {\ep{P_\CDD(f) \big| P_\CDD(f)}}\left\langle s_D^{(m, l)}\right\rangle=  \trace\left[ \CEE^{\otimes n}( S_f)  \ket{P_\CDD(f)}\bra{P_\CDD(f)} \right],
\end{align*}
where $S_f$ is defined in the equation \eqref{80} and $\ket{P_\CDD(f)}$ is the quantum Fourier transform of $\ket{\tilde P_\CDD(f)}$.
The action of  $\CEE^{\otimes n}$ acting on $ S_f$ is equivalent to the channel $\CEE_1^{\otimes n}$ acting on $ S_f$, where 
$\CEE_1(\rho) = (1- \varepsilon)\rho + \frac{\varepsilon}{p} \sum_{\alpha\in \BFF_p} X^\alpha \rho X^{-\alpha} $.
Hence, we have
\begin{align*}
&{\ep{P_\CDD(f) \big| P_\CDD(f)}}\left\langle s_D^{(m, l)}\right\rangle \\
=&\trace\left[  \CEE_1^{\otimes n} (S_f ) \ket{P_\CDD(f)}\bra{P_\CDD(f)} \right]\\
=& \sum_{\vec\alpha\in \BFF_p^n} \left(\varepsilon/p\right)^{|\vec\alpha|} \left(1-(p-1)\varepsilon/p\right)^{1- |\vec\alpha|}   
\bra{P_\CDD(f)} X^{\vec \alpha} S_f X^{-\vec \alpha} \ket{P_\CDD(f)} \\
=& \sum_{\vec\alpha\in \BFF_p^n} \left(\varepsilon/p\right)^{|\vec\alpha|} \left(1-(p-1)\varepsilon/p\right)^{1- |\vec\alpha|}
\frac{1}{p} \sum_{i=1}^m \sum_{v \in F_i} \sum_{a \in \mathbb{F}_p} \omega_p^{-a v} 
\bra{P_\CDD(f)} X^{\vec \alpha} \prod_{j=1}^n Z_j^{a B_{i j}} X^{-\vec \alpha} \ket{P_\CDD(f)}\\
=& \sum_{\vec\alpha\in \BFF_p^n} \left(\varepsilon/p\right)^{|\vec\alpha|} \left(1-(p-1)\varepsilon/p\right)^{1- |\vec\alpha|}
\frac{1}{p} \sum_{i=1}^m \sum_{v \in F_i} \sum_{a \in \mathbb{F}_p} \omega_p^{-a v+ a\ep{\mathbf{b}_i, \vec \alpha}} 
\bra{\widetilde{P}_\CDD(f)}  X^{-a\mathbf{b}_i}  \ket{\widetilde{P}_\CDD(f)}\\
=& \sum_{\vec\alpha\in \BFF_p^n} \left(\varepsilon/p\right)^{|\vec\alpha|} \left(1-(p-1)\varepsilon/p\right)^{1- |\vec\alpha|}   
\frac{1}{p} \sum_{k_1, k_2=0}^{l} \frac{w_{k_1}^* w_{k_2}}{\sqrt{\binom{m}{k_1}\binom{m}{k_2}}}  
\sum_{\substack{ \mathbf{y}_1 \in \CDD_{k_1}\\\mathbf{y}_2 \in \CDD_{k_2}
}} \tilde{g}^*\left(\mathbf{y}_1\right) \tilde{g}\left(\mathbf{y}_2\right)  \\
&\times
\sum_{i=1}^m \sum_{v \in F_i} \sum_{a \in \mathbb{F}_p}
\omega_p^{-a v+a \ep{\mathbf{b}_i, \vec \alpha}}
\left\langle B^T \mathbf{y}_1 \bigg| B^T\left(\mathbf{y}_2-a \mathbf{e}_i\right)\right\rangle,
\end{align*}
where the forth line comes form the equation \eqref{80}, the fifth line comes 
from the fact that $F Z F^{\dagger}=X^{-1}$and $|\widetilde{P}(f)\rangle=F^{\otimes n}|P(f)\rangle$,
and the last line  comes from the equation~\eqref{171}.
Hence, we have 
\begin{align*}
\ep{P_\CDD(f) \big| P_\CDD(f)}
\left\langle s_D^{(m, l)}\right\rangle = &
\sum_{k_1, k_2=0}^{l} {w_{k_1}^* w_{k_2}} 
\bar{A}^{(m, l,\CDD)}_{k_1,k_2},
\end{align*}
where
\begin{align*}
\bar{A}^{(m, l,\CDD)}_{k_1,k_2} = &
  \frac{1}{\sqrt{\binom{m}{k_1}\binom{m}{k_2}}}
\sum_{\vec\alpha\in \BFF_p^n} \left(\varepsilon/p\right)^{|\vec\alpha|} \left(1-(1-p)\varepsilon/p\right)^{1- |\vec\alpha|}   
 \frac{1}{p} \sum_{\substack{\mathbf{y}_1 \in \CDD_{k_1}\\ \mathbf{y}_2 \in \CDD_{k_2}}} 
\tilde{g}^*\left(\mathbf{y}_1\right) \tilde{g} \left(\mathbf{y}_2\right)  \\
& \times
\sum_{i=1}^m \sum_{v \in F_i} \sum_{a \in \mathbb{F}_p}
\omega_p^{-a v+a \ep{\mathbf{b}_i, \vec \alpha}}
\left\langle B^T \mathbf{y}_1 \bigg| B^T\left(\mathbf{y}_2-a \mathbf{e}_i\right)\right\rangle.
\end{align*}
Note that $\left\langle B^T \mathbf{y}_1 \bigg| B^T\left(\mathbf{y}_2-a \mathbf{e}_i\right)\right\rangle= 1$ if $  B^T \mathbf{y}_1 = B^T\left(\mathbf{y}_2-a \mathbf{e}_i\right) $, and zero otherwise.
Therefore, 
\begin{align*}
\bar{A}^{(m, l,\CDD)}_{k_1,k_2} = &
  \frac{1}{\sqrt{\binom{m}{k_1}\binom{m}{k_2}}}
\sum_{\vec\alpha\in \BFF_p^n} \left(\varepsilon/p\right)^{|\vec\alpha|} \left(1-(p-1)\varepsilon/p\right)^{1- |\vec\alpha|}   
  \\
& \times
\frac{1}{p} 
\sum_{i=1}^m 
\sum_{v \in F_i} 
\sum_{a \in \mathbb{F}_p}
\sum_{\substack{(\mathbf{y}_1, \mathbf{y}_2) \in S^{(i,a,\CDD)}_{k_1,k_2}
}} 
\tilde{g}^*\left(\mathbf{y}_1\right) \tilde{g} \left(\mathbf{y}_2\right)
\omega_p^{-a v+a \ep{\mathbf{b}_i, \vec \alpha}} .
\end{align*}
\end{proof}

\begin{lem}\label{10.3}
Let $\bar{A}^{(m, l,E)}$ be defined as in \eqref{250617eq1}. If 
 the sets $F_1,...,F_m$ are chosen independently uniformly at random from the set of all $r$-subsets of $\BFF_p$, then
\begin{align}\label{250617eq2}
\Expect_{F_1,...,F_m} \bar{A}^{(m, l,E)} = \frac{m r}{p} I+\tau_1(B,\varepsilon) \frac{\sqrt{r(p-r)}}{p}  A^{(m, l, d)}  ,
\end{align}
where $\tau_1(B,\varepsilon)$ is defined as \eqref{eq:noise_P}, and $A^{(m, l, d)}$ is defined as in \eqref{111}.
\end{lem}

\begin{proof}
For $0\le k_1,k_2\le l$,
we define the following two subsets of  $S^{(i,a,E)}_{k_1,k_2}$, as defined in \eqref{def:setS},
\begin{align}
S^{(i,a,E,0)}_{k_1,k_2} := &\{(\mathbf{y}_1, \mathbf{y}_2) \in S^{(i,a,E)}_{k_1,k_2}: 
\mathbf{y}_1 - \mathbf{y}_2 + a\mathbf{e}_i=\mathbf{0} \},\\
S^{(i,a,E,1)}_{k_1,k_2} := &\{(\mathbf{y}_1, \mathbf{y}_2) \in S^{(i,a,E)}_{k_1,k_2} : 
\mathbf{y}_1 - \mathbf{y}_2 + a\mathbf{e}_i \neq \mathbf{0} \}.
\end{align}
That is,  $S^{(i,a,E)}_{k_1,k_2} = S^{(i,a,E,0)}_{k_1,k_2} \cup S^{(i,a,E,1)}_{k_1,k_2}$. 
By Lemma \ref{lem10.2}, we have
\begin{align*}
\bar{A}^{(m, l,E)}_{k_1,k_2} = &
  \frac{1}{\sqrt{\binom{m}{k_1}\binom{m}{k_2}}}
\sum_{\vec\alpha\in \BFF_p^n} \left(\varepsilon/p\right)^{|\vec\alpha|} \left(1-(p-1)\varepsilon/p\right)^{1- |\vec\alpha|}   
  \\
& \times
\frac{1}{p} 
\sum_{i=1}^m 
\sum_{v \in F_i} 
\sum_{a \in \mathbb{F}_p}
\sum_{\substack{(\mathbf{y}_1, \mathbf{y}_2) \in S^{(i,a,E,0)}_{k_1,k_2}  }} 
\tilde{g}^*\left(\mathbf{y}_1\right) \tilde{g} \left(\mathbf{y}_2\right)
\omega_p^{-a v+a \ep{\mathbf{b}_i, \vec \alpha}} \\
+&   \frac{1}{\sqrt{\binom{m}{k_1}\binom{m}{k_2}}}
\sum_{\vec\alpha\in \BFF_p^n} \left(\varepsilon/p\right)^{|\vec\alpha|} \left(1-(p-1)\varepsilon/p\right)^{1- |\vec\alpha|}   
  \\
& \times
\frac{1}{p} 
\sum_{i=1}^m 
\sum_{v \in F_i} 
\sum_{a \in \mathbb{F}_p}
\sum_{\substack{(\mathbf{y}_1, \mathbf{y}_2) \in S^{(i,a,E,1)}_{k_1,k_2}  }} 
\tilde{g}^*\left(\mathbf{y}_1\right) \tilde{g} \left(\mathbf{y}_2\right)
\omega_p^{-a v+a \ep{\mathbf{b}_i, \vec \alpha}}.
\end{align*}

Now, for $t=1,2$, and for any $\vec \alpha$, $i$, $k_1$ and $k_2$, we have
\begin{align*}
&\Expect_{F_1,...,F_m}
\sum_{v \in F_i} 
\sum_{a \in \mathbb{F}_p}
\sum_{\substack{(\mathbf{y}_1, \mathbf{y}_2) \in S^{(i,a,E,t)}_{k_1,k_2}  }} 
\tilde{g}^*\left(\mathbf{y}_1\right) \tilde{g} \left(\mathbf{y}_2\right)
\omega_p^{-a v+a \ep{\mathbf{b}_i, \vec \alpha}} \\
=& \sum_{a \in \mathbb{F}_p}
\omega_p^{a \ep{\mathbf{b}_i, \vec \alpha}}
\sum_{\substack{(\mathbf{y}_1, \mathbf{y}_2) \in S^{(i,a,E,t)}_{k_1,k_2}  }}
\Expect_{F_1,...,F_m}
\sum_{v \in F_i}
\tilde{g}^*\left(\mathbf{y}_1\right) \tilde{g} \left(\mathbf{y}_2\right)
\omega_p^{-a v}.
\end{align*}
If  both $a$ and $(\mathbf{y}_1, \mathbf{y}_2) \in S^{(i,a,E,t)}_{k_1,k_2} $ are also fixed, and we assume that $\mathbf{y}_j = (y_{j,1}, ..., y_{j,m})$ for $j=1,2$, we have
\begin{align*}
&\Expect_{F_1,...,F_m}
\sum_{v \in F_i}
\tilde{g}^*\left(\mathbf{y}_1\right) \tilde{g} \left(\mathbf{y}_2\right)
\omega_p^{-a v}\\
=& \Expect_{F_1,...,F_m}
\sum_{v \in F_i}
\prod_{\substack{i_1=1 \\ y_{1,i_1} \neq 0}}^m \tilde{g}_{i_1}\left(y_{1,i_1}\right) ^*
\prod_{\substack{i_2=1 \\ y_{2,i_2} \neq 0}}^m \tilde{g}_{i_2}\left(y_{2,i_2}\right) 
\omega_p^{-a v}\\
=& \frac{1}{p^{(k_1+k_2)/2}}
\Expect_{F_1,...,F_m}
\sum_{v \in F_i}
\prod_{\substack{i_1=1 \\ y_{1,i_1} \neq 0}}^m 
\left(\sum_{x_1 \in \mathbb{F}_p} 
\omega_p^{y_{1,i_1} x_1} g_{i_1}(x_1) \right)^*
\prod_{\substack{i_2=1 \\ y_{2,i_2} \neq 0}}^m 
\left(\sum_{x_2 \in \mathbb{F}_p} 
\omega_p^{y_{2,i_2} x_2} g_{i_2}(x_2) \right)
\omega_p^{-a v}\\
=& \frac{1}{p^{(k_1+k_2)/2} \varphi^{k_1+k_2}}
\Expect_{F_1,...,F_m}
\sum_{v \in F_i}
\prod_{\substack{i_1=1 \\ y_{1,i_1} \neq 0}}^m 
\left(\sum_{x_1 \in \mathbb{F}_p} 
\omega_p^{y_{1,i_1} x_1} f_{i_1}(x_1) \right)^*\\
&\quad\quad\times
\prod_{\substack{i_2=1 \\ y_{2,i_2} \neq 0}}^m 
\left(\sum_{x_2 \in \mathbb{F}_p} 
\omega_p^{y_{2,i_2} x_2} f_{i_2}(x_2) \right)
\omega_p^{-a v}\\
=& \frac{2^{k_1+k_2}}{p^{(k_1+k_2)/2} \varphi^{k_1+k_2}}
\Expect_{F_1,...,F_m}
\sum_{v \in F_i}
\prod_{\substack{i_1=1 \\ y_{1,i_1} \neq 0}}^m 
\left(\sum_{x_1 \in \mathbb{F}_p} 
\omega_p^{y_{1,i_1} x_1} \mathbb{1}_{F_{i_1}}(x_1) \right)^*\\
&\quad\quad\times \prod_{\substack{i_2=1 \\ y_{2,i_2} \neq 0}}^m 
\left(\sum_{x_2 \in \mathbb{F}_p} 
\omega_p^{y_{2,i_2} x_2} \mathbb{1}_{F_{i_2}}(x_2) \right)
\omega_p^{-a v}\\
=& \frac{1}{(r(p-r))^{\frac{k_1+k_2}2}}
\Expect_{F_1,...,F_m}
\sum_{v \in F_i}
\prod_{\substack{i_1=1 \\ y_{1,i_1} \neq 0}}^m 
\left(\sum_{x_1 \in \mathbb{F}_p} 
\omega_p^{y_{1,i_1} x_1} \mathbb{1}_{F_{i_1}}(x_1) \right)^*\\
&\quad\quad\times
\prod_{\substack{i_2=1 \\ y_{2,i_2} \neq 0}}^m 
\left(\sum_{x_2 \in \mathbb{F}_p} 
\omega_p^{y_{2,i_2} x_2} \mathbb{1}_{F_{i_2}}(x_2) \right)
\omega_p^{-a v},
\end{align*}
where  the third equation comes from the equation~\eqref{37},
the forth equation comes from the equation~\eqref{250611eq1}, and the last equation comes from the equation~\eqref{250614eq1}.
Thus, we have
\begin{align*}
&\Expect_{F_1,...,F_m}
\sum_{v \in F_i}
\tilde{g}^*\left(\mathbf{y}_1\right) \tilde{g} \left(\mathbf{y}_2\right)
\omega_p^{-a v}\\
=& \frac{1}{(r(p-r))^{\frac{k_1+k_2}2}}
\Expect_{F_1,...,F_m}
\sum_{v \in F_i}
\prod_{\substack{i_1=1 \\ y_{1,i_1} \neq 0}}^m 
\left(\sum_{x_1 \in F_{i_1}} 
\omega_p^{y_{1,i_1} x_1}  \right)^*
\prod_{\substack{i_2=1 \\ y_{2,i_2} \neq 0}}^m 
\left(\sum_{x_2 \in F_{i_2}} 
\omega_p^{y_{2,i_2} x_2}  \right)
\omega_p^{-a v}\\
=& \frac{1}{(r(p-r))^{\frac{k_1+k_2}2}}
\Expect_{F_1,...,F_m} \frac{1}{r^{m-k_1}}
\sum_{\mathbf{x}_1 \in F_1\times \cdots F_m} 
\omega_p^{- \mathbf{y}_1\cdot \mathbf{x}_1}   
 \frac{1}{r^{m-k_2}}
\sum_{\mathbf{x}_2 \in F_1\times \cdots F_m} 
\omega_p^{\mathbf{y}_2\cdot \mathbf{x}_2} 
\sum_{v \in F_i}
\omega_p^{-a v}\\
=& \frac{1}{(r(p-r))^{\frac{k_1+k_2}2} r^{2m-k_1-k_2}}
\Expect_{F_1,..\hat{F}_i.,F_m}
\left( \sum_{ {x}_{1,1},x_{1,2} \in F_1 } \omega_p^{- {y}_{1,1}  {x}_{1,1}+y_{2,1} x_{2,1}} \right)\times \cdots\\
& 
\times
\left( \sum_{ {x}_{m,1},x_{m,2} \in F_m } \omega_p^{- {y}_{1,m}  {x}_{1,m}+y_{2,m} x_{2,m}} \right)
\times \Expect_{ {F}_i} 
\left( \sum_{ {x}_{1,i},x_{2,i},v \in F_i } \omega_p^{- {y}_{1,i}  {x}_{1,i}+y_{2,i} x_{2,i} -av} \right).
\end{align*}
Due to Lemma~\ref{250612lem1}, we have that
$
\Expect_{ {F}\subseteq [m], |F|=r} 
\sum_{ {x}_{1},x_{2},v \in F } \omega_p^{- {y}_{1,1}  {x}_{1}+y_{2,1} x_{2} -av} 
$
is zero unless ${y}_{1,1}-{y}_{2,1}+a=0$.
Hence,  if $(\mathbf{y}_1, \mathbf{y}_2) \in S^{(i,a,E,1)}_{k_1,k_2} $,  we have
\begin{align}
\Expect_{F_1,...,F_m}
\sum_{v \in F_i}
\tilde{g}^*\left(\mathbf{y}_1\right) \tilde{g} \left(\mathbf{y}_2\right)
\omega_p^{-a v} =0.
\end{align}
In addition,
by Lemma \ref{250613lem2},
when $(\mathbf{y}_1, \mathbf{y}_2) \in S^{(i,a,E,0)}_{k_1,k_2} $ with $ \mathbf{y}_1= \mathbf{y}_2$ (i.e., $a=0$ and $k_1=k_2$),
we have
\begin{align*}
&\Expect_{F_1,...,F_m}
\sum_{v \in F_i}
\tilde{g}^*\left(\mathbf{y}_1\right) \tilde{g} \left(\mathbf{y}_2\right)
\omega_p^{-a v} \\
= & \frac{1}{(r(p-r))^{\frac{k_1+k_2}2}r^{2m-k_1-k_2}}
\Expect_{F_1,..\hat{F}_i.,F_m}
\left( \sum_{ {x}_{1,1},x_{1,2} \in F_1 } \omega_p^{- {y}_{1,1}  {x}_{1,1}+y_{2,1} x_{2,1}} \right)\times\cdots\\
& \times
\left( \sum_{ {x}_{m,1},x_{m,2} \in F_m } \omega_p^{- {y}_{1,m}  {x}_{1,m}+y_{2,m} x_{2,m}} \right)
\times \Expect_{ {F}_i} 
\left( \sum_{ {x}_{1,i},x_{2,i},v \in F_i } \omega_p^{- {y}_{1,i}  {x}_{1,i}+y_{2,i} x_{2,i} -av} \right)\\
=& \frac{1}{(r(p-r))^{k_1} r^{2m-2k_1}} 
\left( r-\frac{r(r-1)}{p-1} \right)^{k_1} (r^2)^{m-k_1} r \\
=& \frac{r}{(p-1)^{k_1}}.
\end{align*}

Next,
we assume $(\mathbf{y}_1, \mathbf{y}_2) \in S^{(i,a,E,0)}_{k_1,k_2} $ with $ \mathbf{y}_1\neq \mathbf{y}_2$ (i.e., 
$\mathbf{y}_1 - \mathbf{y}_2 + a\mathbf{e}_i=\mathbf{0} $ an $a\neq 0$).  Then we have $k_1=k_2 \pm 1$ or $k_1=k_2$.
If $k_1=k_2 - 1$, then$ y_{2,i} =a$ and $y_{1,i}=0$. Again, by Lemma \ref{250613lem2},
we have
\begin{align*}
&\Expect_{F_1,...,F_m}
\sum_{v \in F_i}
\tilde{g}^*\left(\mathbf{y}_1\right) \tilde{g} \left(\mathbf{y}_2\right)
\omega_p^{-a v} \\
= & \frac{1}{(r(p-r))^{\frac{k_1+k_2}2} r^{2m-k_1-k_2}}
\Expect_{F_1,..\hat{F}_i.,F_m}
\left( \sum_{ {x}_{1,1},x_{1,2} \in F_1 } \omega_p^{- {y}_{1,1}  {x}_{1,1}+y_{2,1} x_{2,1}} \right) \times\cdots\\
&\times
\left( \sum_{ {x}_{m,1},x_{m,2} \in F_m } \omega_p^{- {y}_{1,m}  {x}_{1,m}+y_{2,m} x_{2,m}} \right)
\times \Expect_{ {F}_i} 
\left( \sum_{ {x}_{1,i},x_{2,i},v \in F_i } \omega_p^{- {y}_{1,i}  {x}_{1,i}+y_{2,i} x_{2,i} -av} \right)\\
=& \frac{1}{(r(p-r))^{\frac{k_1+k_2}2} r^{2m-k_1-k_2}} 
\left( r-\frac{r(r-1)}{p-1} \right)^{k_1+1} r (r^2)^{m-k_1-1} \\
=& \frac{\sqrt{r(p-r)}}{(p-1)^{k_1+1}}.
\end{align*}
Similarly,
if $k_2=k_1 - 1$,
\begin{align*}
\Expect_{F_1,...,F_m}
\sum_{v \in F_i}
\tilde{g}^*\left(\mathbf{y}_1\right) \tilde{g} \left(\mathbf{y}_2\right)
\omega_p^{-a v} 
=& \frac{\sqrt{r(p-r)}}{(p-1)^{k_2+1}}.
\end{align*}
Finally,
if $k_1=k_2$,
then $ y_{1,i} ,y_{2,i}\neq 0$ and $ y_{1,i} -y_{2,i}+a=0$. By Lemmas~\ref{250613lem2} and~\ref{250613lem3}, 
we have
\begin{align*}
&\Expect_{F_1,...,F_m}
\sum_{v \in F_i}
\tilde{g}^*\left(\mathbf{y}_1\right) \tilde{g} \left(\mathbf{y}_2\right)
\omega_p^{-a v} \\
= & \frac{1}{(r(p-r))^{\frac{k_1+k_2}2} r^{2m-k_1-k_2}}
\Expect_{F_1,..\hat{F}_i.,F_m}
\left( \sum_{ {x}_{1,1},x_{1,2} \in F_1 } \omega_p^{- {y}_{1,1}  {x}_{1,1}+y_{2,1} x_{2,1}} \right)\times\cdots\\
&\times
\left( \sum_{ {x}_{m,1},x_{m,2} \in F_m } \omega_p^{- {y}_{1,m}  {x}_{1,m}+y_{2,m} x_{2,m}} \right)
\times \Expect_{ {F}_i} 
\left( \sum_{ {x}_{1,i},x_{2,i},v \in F_i } \omega_p^{- {y}_{1,i}  {x}_{1,i}+y_{2,i} x_{2,i} -av} \right)\\
=& \frac{1}{(r(p-r))^{k_1} r^{2m-2k_1}} 
\left( r-\frac{r(r-1)}{p-1} \right)^{k_1-1} (r^2)^{m-k_1}
\frac{r(p-r)(p-2r)}{(p-1)(p-2)}\\
=& \frac{p-2r}{(p-1)^{k_1} (p-2)} .
\end{align*}

Therefore,
when $k_1=k_2$, we have
\begin{align*}
&\Expect_{F_1,...,F_m}\bar{A}^{(m, l,E)}_{k_1,k_2} \\
= &   \frac{1}{\sqrt{\binom{m}{k_1}\binom{m}{k_2}}}
\sum_{\vec\alpha\in \BFF_p^n} 
\left(\varepsilon/p\right)^{|\vec\alpha|} \left(1-(p-1)\varepsilon/p\right)^{1- |\vec\alpha|}   \frac{1}{p} 
\sum_{i=1}^m 
\sum_{a \in \mathbb{F}_p}
\omega_p^{a \ep{\mathbf{b}_i, \vec \alpha}}
  \\
& \quad\times
\Expect_{F_1,...,F_m}
\sum_{v \in F_i} 
\sum_{\substack{(\mathbf{y}_1, \mathbf{y}_2) \in S^{(i,a,E,0)}_{k_1,k_2}  }} 
\tilde{g}^*\left(\mathbf{y}_1\right) \tilde{g} \left(\mathbf{y}_2\right)
\omega_p^{-a v}\\
=& \frac{1}{\sqrt{\binom{m}{k_1}\binom{m}{k_2}}}
\sum_{\vec\alpha\in \BFF_p^n} 
\left(\varepsilon/p\right)^{|\vec\alpha|} \left(1-(p-1)\varepsilon/p\right)^{1- |\vec\alpha|}   \frac{1}{p} 
\sum_{i=1}^m 
\sum_{a \in \mathbb{F}_p}
\omega_p^{a \ep{\mathbf{b}_i, \vec \alpha}}
  \\
& \qquad \times
\Expect_{F_1,...,F_m}
\sum_{v \in F_i} 
\sum_{\substack{(\mathbf{y}_1, \mathbf{y}_2) \in E_{k_1} \times E_{k_2}\\ \mathbf{y}_1 - \mathbf{y}_2 + a\mathbf{e}_i = \mathbf{0} }} 
\tilde{g}^*\left(\mathbf{y}_1\right) \tilde{g} \left(\mathbf{y}_2\right)
\omega_p^{-a v}\\
=&\frac{1}{ \binom{m}{k_1} }
\sum_{\vec\alpha\in \BFF_p^n} 
\left(\varepsilon/p\right)^{|\vec\alpha|} \left(1-(p-1)\varepsilon/p\right)^{1- |\vec\alpha|}   
\frac{1}{p} 
\sum_{i=1}^m \\
& 
\left(
\sum_{\substack{(\mathbf{y}_1, \mathbf{y}_2) \in E_{k_1} \times E_{k_2}\\ \mathbf{y}_1 = \mathbf{y}_2  }}  
\frac{r}{(p-1)^{k_1}} 
 +  \sum_{a \in \mathbb{F}_p^*} \sum_{\substack{(\mathbf{y}_1, \mathbf{y}_2) \in E_{k_1} \times E_{k_2}\\ \mathbf{y}_1 - \mathbf{y}_2 + a\mathbf{e}_i = \mathbf{0} }} 
\omega_p^{a \ep{\mathbf{b}_i, \vec \alpha}}
 \frac{p-2r}{(p-1)^{k_1} (p-2)}
\right)\\
=&  \frac{1}{ \binom{m}{k_1} } \frac{m}{p} \binom{m}{k_1} (p-1)^{k_1}\frac{r}{(p-1)^{k_1}}  +  \frac{p-2r}{(p-1)^{k_1} (p-2)}
\\
& \quad \times
\frac{m}{p}  
\binom{m-1}{k_1-1}(p-1)^{k_1-1}(p-2)
\frac{1}{ \binom{m}{k_1} } 
\sum_{\vec\alpha\in \BFF_p^n} 
\left(\varepsilon/p\right)^{|\vec\alpha|} \left(1-(p-1)\varepsilon/p\right)^{1- |\vec\alpha|}   
\Expect_{i\in [m]}
\sum_{a \in \mathbb{F}_p^*} 
\omega_p^{a \ep{\mathbf{b}_i, \vec \alpha}}.
\end{align*}

Recall that we use $Q(t,i)$ to denote the probability of $\ep{\Buu, \mathbf{b}_i}=0$ when $\Buu $ is uniformly chosen from the set $\{\Buu \in \BFF_p^n: |\Buu| =t\}$, hence
\begin{align*}
&\Expect_{F_1,...,F_m}\bar{A}^{(m, l,E)}_{k_1,k_1} \\
=&   \frac{mr}{p}    + \frac{(p-2r)k_1 }{p (p-1)}  
\sum_{\vec\alpha\in \BFF_p^n} 
\left(\varepsilon/p\right)^{|\vec\alpha|} \left(1-(p-1)\varepsilon/p\right)^{1- |\vec\alpha|} 
\Expect_{i\in [m]} (p-1)
\left( Q(|\vec\alpha|, i) - (1- Q(|\vec\alpha|, i) )\frac 1{p-1} )
\right) \\
=& \frac{mr}{p} + \frac{(p-2r)k_1 }{p } \tau_1(B,\varepsilon),
\end{align*}
where the last line comes from Lemma~\ref{250627lem2}.

In addition, when $k_1=k_2-1$, we have
\begin{align*}
&\Expect_{F_1,...,F_m}\bar{A}^{(m, l,E)}_{k_1,k_2} \\
= &   \frac{1}{\sqrt{\binom{m}{k_1}\binom{m}{k_2}}}
\sum_{\vec\alpha\in \BFF_p^n} 
\left(\varepsilon/p\right)^{|\vec\alpha|} \left(1-(p-1)\varepsilon/p\right)^{1- |\vec\alpha|}   \frac{1}{p} 
\sum_{i=1}^m 
\sum_{a \in \mathbb{F}_p}
\omega_p^{a \ep{\mathbf{b}_i, \vec \alpha}}
  \\
& \qquad\times
\Expect_{F_1,...,F_m}
\sum_{v \in F_i} 
\sum_{\substack{(\mathbf{y}_1, \mathbf{y}_2) \in S^{(i,a,E,0)}_{k_1,k_2}  }} 
\tilde{g}^*\left(\mathbf{y}_1\right) \tilde{g} \left(\mathbf{y}_2\right)
\omega_p^{-a v} \\
= &   \frac{1}{\sqrt{\binom{m}{k_1}\binom{m}{k_2}}}
\sum_{\vec\alpha\in \BFF_p^n} 
\left(\varepsilon/p\right)^{|\vec\alpha|} \left(1-(p-1)\varepsilon/p\right)^{1- |\vec\alpha|}   
\frac{1}{p} 
\sum_{i=1}^m 
\sum_{a \in \mathbb{F}_p^*}
\omega_p^{a \ep{\mathbf{b}_i, \vec \alpha}}
\sum_{\substack{(\mathbf{y}_1, \mathbf{y}_2) \in E_{k_1} \times E_{k_2}\\ \mathbf{y}_1 - \mathbf{y}_2 + a\mathbf{e}_i = \mathbf{0} }}
\frac{\sqrt{r(p-r)}}{(p-1)^{k_1+1}}\\
=& \frac{1}{\sqrt{\binom{m}{k_1}\binom{m}{k_2}}}
\frac{m\sqrt{r(p-r)}}{p(p-1)^{k_1+1}}
\binom{m-1}{k_1}(p-1)^{k_1}
\sum_{\vec\alpha\in \BFF_p^n} 
\left(\varepsilon/p\right)^{|\vec\alpha|} \left(1-(p-1)\varepsilon/p\right)^{1- |\vec\alpha|} 
\Expect_{i\in [m]} 
\sum_{a \in \mathbb{F}_p^*}
\omega_p^{a \ep{\mathbf{b}_i, \vec \alpha}}\\
=& \sqrt{(k_1+1)(m-k_1)} 
\frac{\sqrt{r(p-r)}}{p} 
\tau_1(B,\varepsilon).
\end{align*}

Similarly, when $k_1=k_2+1$, we also have 
\begin{align*}
&\Expect_{F_1,...,F_m}\bar{A}^{(m, l,E)}_{k_1,k_2} 
=\sqrt{(k_2+1)(m-k_2)} 
\frac{\sqrt{r(p-r)}}{p} 
\tau_1(B,\varepsilon).
\end{align*}
Therefore, the equation~\eqref{250617eq2} holds.
\end{proof}

In the following we will compute the expectation of $ \bar{A}^{(m, l,\CDD)}$ when the sets $F_1,...,F_m$ are independently uniformly chosen from all possible subsets of $\{1,2,...,p\}$ of size $r$.
We first deal with the state $\ket{P(f)}$ instead of $\ket{P_\CDD(f)}$.
Similar to Theorem \ref{thm:main_2},
we have the following result.

\begin{lem}\label{lem10.2}
Let $f(\mathbf{x})=\sum_{i=1}^m f_i\left(\sum_{j=1}^n B_{i j} x_j\right)$ be a \maxlinsat{} objective function with matrix $B \in \BFF_p^{m \times n}$ for a prime $p$ and positive integers $m$ and $n$ such that $m>n$. 
Suppose that $\left|f_i^{-1}(+1)\right|=r$ for some $r \in\{1, \ldots, p-1\}$. 
Let 
$P(f) $ be the degree-$l$ polynomial determined by coefficients $w_0,...,w_l$ such that the perfect DQI state $\ket{P(f)}$ satisfies \eqref{51} (note that $\ket{P(f)}$ is not normalized).
Let $\left\langle s_E^{(m, l)}\right\rangle$ be the expected number of satisfied constraints for the symbol string obtained upon measuring the errored imperfect DQI state $\CEE^{\otimes n} (\ket{P(f)}\bra{P(f)})$ in the computational basis. 
Then
$$
\left\langle s_E^{(m, l)}\right\rangle=  \frac{\mathbf{w}^{\dagger} \bar{A}^{(m, l,E)} \mathbf{w}}{\ep{P(f)\big|P(f)}} ,
$$
where $\bar{A}^{(m, l,E)}$ is the $(l+1) \times(l+1)$ symmetric matrix defined by
\begin{align}\label{250617eq1}
\begin{aligned}
\bar{A}^{(m, l,E)}_{k_1,k_2} = &
  \frac{1}{\sqrt{\binom{m}{k_1}\binom{m}{k_2}}}
\sum_{\vec\alpha\in \BFF_p^n} \left(\varepsilon/p\right)^{|\vec\alpha|} \left(1-(p-1)\varepsilon/p\right)^{1- |\vec\alpha|}   
  \\
& \times
\frac{1}{p} 
\sum_{i=1}^m 
\sum_{v \in F_i} 
\sum_{a \in \mathbb{F}_p}
\sum_{\substack{(\mathbf{y}_1, \mathbf{y}_2) \in S^{(i,a,E)}_{k_1,k_2}
}} 
\tilde{g}^*\left(\mathbf{y}_1\right) \tilde{g} \left(\mathbf{y}_2\right)
\omega_p^{-a v+a \ep{\mathbf{b}_i, \vec \alpha}}
\end{aligned}
\end{align} 
for $0\le k_1,k_2\le l$, and 
\begin{align}
S^{(i,a,E)}_{k_1,k_2} = \{(\mathbf{y}_1, \mathbf{y}_2) \in E_{k_1} \times E_{k_2} : 
B^T(\mathbf{y}_1 - \mathbf{y}_2 + a\mathbf{e}_i)=\mathbf{0} \}.
\end{align}
\end{lem}

Now, let us denote 
\begin{align}\label{250621eq1}
T^{(i,a,\CFF)}_{k_1, k_2} 
:= \left\{ (\mathbf{y}_1, \mathbf{y}_2)\in E_{k_1} \times \CFF_{k_2} \cup \CFF_{k_1}\times E_{k_2}: \mathbf{y}_1 - \mathbf{y}_2 + a\mathbf{e}_i=\mathbf{0} \right\}.
\end{align}

\begin{lem}\label{10.6}
Let $\bar{A}^{(m, l,\CDD)}$ be defined as in \eqref{177}.
If the sets $F_1,...,F_m$ are chosen independently uniformly at random from the set of all $r$-subsets of $\BFF_p$,
then we have
\begin{align}\label{250621eq3}
\Expect_{F_1,...,F_m} \bar{A}^{(m, l,\CDD)} = \Expect_{F_1,...,F_m} \bar{A}^{(m, l,E)} - D^{(m,l,\CFF)} ,
\end{align}
where $\bar{A}^{(m, l,E)}$ is defined as in \eqref{250617eq1},
and $D^{(m,l,\CFF)}$ is the $(l+1)\times (l+1)$ symmetric matrix whose $(k_1,k_2)$-entry $D^{(m,l,\CFF)}_{k_1,k_2}$ satisfies that,
$D^{(m,l,\CFF)}_{k_1,k_2}=0$ when $|k_1- k_2|\ge 2$,
\begin{align}
D^{(m,l,\CFF)}_{k_1,k_1+1} = \frac{1}{\sqrt{\binom{m}{k_1}\binom{m}{k_1+1}}}
\frac{\sqrt{r(p-r)}}{(p-1)^{k_1+1}}
\frac{1}{p} 
\sum_{i=1}^m
\sum_{a \in \mathbb{F}_p^*}
\left| T^{(i,a,\CFF)}_{k_1,k_1+1} \right|
\tau(B,\varepsilon,i),
\end{align}
when $k_1=k_2-1$,
and 
\begin{align*}
D^{(m,l,\CFF)}_{k_1,k_1} =\frac{mr}{p}
\gamma_{k_1}
+ 
\frac{1}{ \binom{m}{k_1} }
\frac{1}{p} \frac{p-2r}{(p-1)^{k_1} (p-2)}
\sum_{i=1}^m 
\sum_{a \in \mathbb{F}_p^*}
\left( \left|T^{(i,a,\CFF)}_{k_1,k_1}\right| - \left|\CFF_{k_1}\right|\right) \tau(B, \varepsilon,i)
\end{align*}
when $k_1=k_2$,
where $T^{(i,a,\CFF)}_{k_1, k_2}$ is defined as in \eqref{250621eq1}.
\end{lem}

\begin{proof}
For $0\le k_1,k_2\le l$,
we define
\begin{align}
S^{(i,a,\CDD,0)}_{k_1,k_2} := &\{(\mathbf{y}_1, \mathbf{y}_2) \in S^{(i,a,\CDD)}_{k_1,k_2}: 
\mathbf{y}_1 - \mathbf{y}_2 + a\mathbf{e}_i=\mathbf{0} \}\\
S^{(i,a,\CDD,1)}_{k_1,k_2} := &\{(\mathbf{y}_1, \mathbf{y}_2) \in S^{(i,a,\CDD)}_{k_1,k_2} : 
\mathbf{y}_1 - \mathbf{y}_2 + a\mathbf{e}_i \neq \mathbf{0} \}.
\end{align}
Hence, $S^{(i,a,\CDD)}_{k_1,k_2} = S^{(i,a,\CDD,0)}_{k_1,k_2} \cup S^{(i,a,\CDD,1)}_{k_1,k_2}$.
Similar to the proof of Lemma \ref{10.3}, 
if $(\mathbf{y}_1, \mathbf{y}_2) \in S^{(i,a,\CDD,1)}_{k_1,k_2} $, we have
\begin{align}
\Expect_{F_1,...,F_m}
\sum_{v \in F_i}
\tilde{g}^*\left(\mathbf{y}_1\right) \tilde{g} \left(\mathbf{y}_2\right)
\omega_p^{-a v} =0.
\end{align}
Thus,
\begin{align*}
\bar{A}^{(m, l,\CDD)}_{k_1,k_2} = &
  \frac{1}{\sqrt{\binom{m}{k_1}\binom{m}{k_2}}}
\sum_{\vec\alpha\in \BFF_p^n} \left(\varepsilon/p\right)^{|\vec\alpha|} \left(1-(p-1)\varepsilon/p\right)^{1- |\vec\alpha|}   
  \\
& \qquad \times
\frac{1}{p} 
\sum_{i=1}^m 
\sum_{v \in F_i} 
\sum_{a \in \mathbb{F}_p}
\sum_{\substack{(\mathbf{y}_1, \mathbf{y}_2) \in S^{(i,a,\CDD,0)}_{k_1,k_2}  }} 
\tilde{g}^*\left(\mathbf{y}_1\right) \tilde{g} \left(\mathbf{y}_2\right)
\omega_p^{-a v+a \ep{\mathbf{b}_i, \vec \alpha}},
\end{align*}
which will become $\bar{A}^{(m, l,E)}_{k_1,k_2}$ if we replace $S^{(i,a,\CDD,0)}_{k_1,k_2}$ by $S^{(i,a,E,0)}_{k_1,k_2}$ in the above equation.
Note that $S^{(i,a,E,0)}_{k_1,k_2} = S^{(i,a,\CDD,0)}_{k_1,k_2} \cup T^{(i,a,\CFF)}_{k_1,k_2}$,
hence
\begin{align*}
D^{(m, l,\CFF)}_{k_1,k_2} = &
\Expect_{F_1,...,F_m} \bar{A}^{(m, l,E)}_{k_1,k_2} - \Expect_{F_1,...,F_m} \bar{A}^{(m, l,\CDD)}_{k_1,k_2} \\
=& 
\frac{1}{\sqrt{\binom{m}{k_1}\binom{m}{k_2}}}
\sum_{\vec\alpha\in \BFF_p^n} \left(\varepsilon/p\right)^{|\vec\alpha|} \left(1-(p-1)\varepsilon/p\right)^{1- |\vec\alpha|}   
  \\
& \quad \times
\Expect_{F_1,...,F_m} 
\frac{1}{p} 
\sum_{i=1}^m 
\sum_{v \in F_i} 
\sum_{a \in \mathbb{F}_p}
\sum_{\substack{(\mathbf{y}_1, \mathbf{y}_2) \in T^{(i,a,\CFF)}_{k_1,k_2}  }} 
\tilde{g}^*\left(\mathbf{y}_1\right) \tilde{g} \left(\mathbf{y}_2\right)
\omega_p^{-a v+a \ep{\mathbf{b}_i, \vec \alpha}}.
\end{align*}
Similar to the proof of Lemma \ref{10.3},
when $(\mathbf{y}_1, \mathbf{y}_2) \in T^{(i,a,\CFF)}_{k_1,k_2} $ with $ \mathbf{y}_1= \mathbf{y}_2$ (i.e., $a=0$ and $k_1=k_2$),
we have
\begin{align*}
\Expect_{F_1,...,F_m}
\sum_{v \in F_i}
\tilde{g}^*\left(\mathbf{y}_1\right) \tilde{g} \left(\mathbf{y}_2\right)
\omega_p^{-a v}  
=& \frac{r}{(p-1)^{k_1}}.
\end{align*}
When $(\mathbf{y}_1, \mathbf{y}_2) \in T^{(i,a,\CFF)}_{k_1,k_2} $ with $ \mathbf{y}_1\neq \mathbf{y}_2$ (i.e., 
$\mathbf{y}_1 - \mathbf{y}_2 + a\mathbf{e}_i=\mathbf{0} $ an $a\neq 0$), 
there are two possibilities: $k_1=k_2 \pm 1$ or $k_1=k_2$.
If $k_1=k_2 - 1$, we must have $ y_{2,i} =a$ and $y_{1,i}=0$, and  then
\begin{align*}
\Expect_{F_1,...,F_m}
\sum_{v \in F_i}
\tilde{g}^*\left(\mathbf{y}_1\right) \tilde{g} \left(\mathbf{y}_2\right)
\omega_p^{-a v}
=& \frac{\sqrt{r(p-r)}}{(p-1)^{k_1+1}}.
\end{align*}
If $k_1=k_2 + 1$, we must have $ y_{1,i} =-a$ and $y_{2,i}=0$, and then
\begin{align*}
\Expect_{F_1,...,F_m}
\sum_{v \in F_i}
\tilde{g}^*\left(\mathbf{y}_1\right) \tilde{g} \left(\mathbf{y}_2\right)
\omega_p^{-a v}
=& \frac{\sqrt{r(p-r)}}{(p-1)^{k_1}}.
\end{align*}
If $k_1=k_2 $, we must have $ y_{1,i} -y_{2,i}+a=0$,
and then
\begin{align*}
\Expect_{F_1,...,F_m}
\sum_{v \in F_i}
\tilde{g}^*\left(\mathbf{y}_1\right) \tilde{g} \left(\mathbf{y}_2\right)
\omega_p^{-a v}  
=& \frac{p-2r}{(p-1)^{k_1} (p-2)}.
\end{align*}
Therefore,
when $k_1=k_2$, we have
\begin{align*}
&D^{(m, l,\CFF)}_{k_1,k_2}\\
=&\frac{1}{\sqrt{\binom{m}{k_1}\binom{m}{k_2}}}
\sum_{\vec\alpha\in \BFF_p^n} \left(\varepsilon/p\right)^{|\vec\alpha|} \left(1-(p-1)\varepsilon/p\right)^{1- |\vec\alpha|}   
\Expect_{F_1,...,F_m} 
\frac{1}{p} 
\sum_{i=1}^m 
\sum_{v \in F_i} 
\sum_{a \in \mathbb{F}_p}\\
& \quad
\sum_{\substack{(\mathbf{y}_1, \mathbf{y}_2) \in T^{(i,a,\CFF)}_{k_1,k_2}  }} 
\tilde{g}^*\left(\mathbf{y}_1\right) \tilde{g} \left(\mathbf{y}_2\right)
\omega_p^{-a v+a \ep{\mathbf{b}_i, \vec \alpha}}\\
=& \frac{1}{ \binom{m}{k_1} }
\sum_{\vec\alpha\in \BFF_p^n} \left(\varepsilon/p\right)^{|\vec\alpha|} \left(1-(p-1)\varepsilon/p\right)^{1- |\vec\alpha|}   
\Expect_{F_1,...,F_m} 
\frac{1}{p} 
\sum_{i=1}^m 
\sum_{v \in F_i} 
\sum_{a \in \mathbb{F}_p} \\
& \quad
\left(
\sum_{\substack{(\mathbf{y}_1, \mathbf{y}_2) \in T^{(i,a,\CFF)}_{k_1,k_1} \\ \mathbf{y}_1=  \mathbf{y}_2  }} 
\tilde{g}^*\left(\mathbf{y}_1\right) \tilde{g} \left(\mathbf{y}_2\right)
\omega_p^{-a v+a \ep{\mathbf{b}_i, \vec \alpha}} +
\sum_{\substack{(\mathbf{y}_1, \mathbf{y}_2) \in T^{(i,a,\CFF)}_{k_1,k_1} \\ \mathbf{y}_1\neq  \mathbf{y}_2  }} 
\tilde{g}^*\left(\mathbf{y}_1\right) \tilde{g} \left(\mathbf{y}_2\right)
\omega_p^{-a v+a \ep{\mathbf{b}_i, \vec \alpha}}\right)\\
=& 
\frac{1}{ \binom{m}{k_1} }
\sum_{\vec\alpha\in \BFF_p^n} \left(\varepsilon/p\right)^{|\vec\alpha|} \left(1-(p-1)\varepsilon/p\right)^{1- |\vec\alpha|}  
\frac{1}{p} 
\sum_{i=1}^m
\sum_{\substack{ \mathbf{y}_1  \in \CFF_{k_1}   }}
\frac{r}{(p-1)^{k_1}} \\
& + \frac{1}{ \binom{m}{k_1} }
\sum_{\vec\alpha\in \BFF_p^n} \left(\varepsilon/p\right)^{|\vec\alpha|} \left(1-(p-1)\varepsilon/p\right)^{1- |\vec\alpha|}   
\frac{1}{p} 
\sum_{i=1}^m 
\sum_{a \in \mathbb{F}_p^*}
\omega_p^{ a \ep{\mathbf{b}_i, \vec \alpha}} 
\sum_{\substack{(\mathbf{y}_1, \mathbf{y}_2) \in T^{(i,a,\CFF)}_{k_1,k_1} \\ \mathbf{y}_1\neq  \mathbf{y}_2  }} 
\frac{p-2r}{(p-1)^{k_1} (p-2)}\\
=& \frac{1}{ \binom{m}{k_1} } \frac{mr}{p(p-1)^{k_1}}
\sum_{\substack{ \mathbf{y}_1  \in \CFF_{k_1}   }} 1
+ \frac{1}{ \binom{m}{k_1} }
\frac{1}{p} \frac{p-2r}{(p-1)^{k_1} (p-2)}
\sum_{i=1}^m 
\sum_{a \in \mathbb{F}_p^*}
\sum_{\substack{(\mathbf{y}_1, \mathbf{y}_2) \in T^{(i,a,\CFF)}_{k_1,k_1} \\ \mathbf{y}_1\neq  \mathbf{y}_2  }}
\\
& \quad 
\sum_{\vec\alpha\in \BFF_p^n} \left(\varepsilon/p\right)^{|\vec\alpha|} \left(1-(p-1)\varepsilon/p\right)^{1- |\vec\alpha|}
\omega_p^{ a \ep{\mathbf{b}_i, \vec \alpha}}\\
=& \frac{1}{ \binom{m}{k_1} } \frac{mr}{p(p-1)^{k_1}}
\sum_{\substack{ \mathbf{y}_1  \in \CFF_{k_1}   }} 1
+ 
\frac{1}{ \binom{m}{k_1} }
\frac{1}{p} \frac{p-2r}{(p-1)^{k_1} (p-2)}
\sum_{i=1}^m 
\sum_{a \in \mathbb{F}_p^*}
\sum_{\substack{(\mathbf{y}_1, \mathbf{y}_2) \in T^{(i,a,\CFF)}_{k_1,k_1} \\ \mathbf{y}_1\neq  \mathbf{y}_2  }}
\\
& \quad 
\sum_{\vec\alpha\in \BFF_p^n} \left(\frac{\varepsilon}{p}\right)^{|\vec\alpha|} \left(1-\frac{p-1}{p}\varepsilon\right)^{1- |\vec\alpha|}
\left(
Q(|\vec\alpha|, i) - (1-Q(|\vec\alpha|, i)) \frac{1}{p-1}
\right) \\
=& \frac{1}{ \binom{m}{k_1} } \frac{mr}{p(p-1)^{k_1}}
\left|  \CFF_{k_1} \right|
+ 
\frac{1}{ \binom{m}{k_1} }
\frac{1}{p} \frac{p-2r}{(p-1)^{k_1} (p-2)}
\sum_{i=1}^m 
\sum_{a \in \mathbb{F}_p^*}
\left( \left|T^{(i,a,\CFF)}_{k_1,k_1}\right| - \left|\CFF_{k_1}\right|\right)
\\
& \quad 
\sum_{\vec\alpha\in \BFF_p^n} \left(\varepsilon/p\right)^{|\vec\alpha|} \left(1-(p-1)\varepsilon/p\right)^{1- |\vec\alpha|}
\left( \frac{p}{p-1}
Q(|\vec\alpha|, i) -  \frac{1}{p-1}
\right)\\
=&  \frac{mr}{p}
\gamma_{k_1}
+ 
\frac{1}{ \binom{m}{k_1} }
\frac{1}{p} \frac{p-2r}{(p-1)^{k_1} (p-2)}
\sum_{i=1}^m 
\sum_{a \in \mathbb{F}_p^*}
\left( \left|T^{(i,a,\CFF)}_{k_1,k_1}\right| - \left|\CFF_{k_1}\right|\right) \tau(B, \varepsilon,i).
\end{align*}
In addition, when $k_1=k_2-1$,
\begin{align*}
&D^{(m, l,\CFF)}_{k_1,k_2}\\
=&\frac{1}{\sqrt{\binom{m}{k_1}\binom{m}{k_2}}}
\sum_{\vec\alpha\in \BFF_p^n} \left(\varepsilon/p\right)^{|\vec\alpha|} \left(1-(p-1)\varepsilon/p\right)^{1- |\vec\alpha|}   
\Expect_{F_1,...,F_m} 
\frac{1}{p} 
\sum_{i=1}^m 
\sum_{v \in F_i} 
\sum_{a \in \mathbb{F}_p}\\
& \quad
\sum_{\substack{(\mathbf{y}_1, \mathbf{y}_2) \in T^{(i,a,\CFF)}_{k_1,k_2}  }} 
\tilde{g}^*\left(\mathbf{y}_1\right) \tilde{g} \left(\mathbf{y}_2\right)
\omega_p^{-a v+a \ep{\mathbf{b}_i, \vec \alpha}}\\
=& \frac{1}{\sqrt{\binom{m}{k_1}\binom{m}{k_2}}}
\sum_{\vec\alpha\in \BFF_p^n} \left(\varepsilon/p\right)^{|\vec\alpha|} \left(1-(p-1)\varepsilon/p\right)^{1- |\vec\alpha|}
\frac{1}{p} 
\sum_{i=1}^m
\sum_{a \in \mathbb{F}_p^*}
\sum_{\substack{(\mathbf{y}_1, \mathbf{y}_2) \in T^{(i,a,\CFF)}_{k_1,k_2}  }}
\frac{\sqrt{r(p-r)}}{(p-1)^{k_1+1}}
\omega_p^{a \ep{\mathbf{b}_i, \vec \alpha}}\\
=& \frac{1}{\sqrt{\binom{m}{k_1}\binom{m}{k_1+1}}}
\frac{\sqrt{r(p-r)}}{(p-1)^{k_1+1}}
\frac{1}{p} 
\sum_{i=1}^m
\sum_{a \in \mathbb{F}_p^*}
\left| T^{(i,a,\CFF)}_{k_1,k_1+1} \right|
\sum_{\vec\alpha\in \BFF_p^n} \left(\varepsilon/p\right)^{|\vec\alpha|} \left(1-(p-1)\varepsilon/p\right)^{1- |\vec\alpha|} \omega_p^{a \ep{\mathbf{b}_i, \vec \alpha}}\\
=& \frac{1}{\sqrt{\binom{m}{k_1}\binom{m}{k_1+1}}}
\frac{\sqrt{r(p-r)}}{(p-1)^{k_1+1}}
\frac{1}{p} 
\sum_{i=1}^m
\sum_{a \in \mathbb{F}_p^*}
\left| T^{(i,a,\CFF)}_{k_1,k_1+1} \right|
\tau(B,\varepsilon,i).
\end{align*}
The case where $k_1=k_2+1$ can be handled similarly.
When $|k_1-k_2| \ge 2$, $T^{(i,a,\CFF)}_{k_1,k_1+1} = \emptyset$,
which implies that  $D_{k_1,k_2}^{(m,l,\CFF)} =0$.
\end{proof}

Now, let us consider the case where $p=2$ and $r=1$.
In this case, we have
$g_i(y) = \pm \frac{1}{\sqrt 2}$ for all $i$ and $y$ ,
and thus $\tilde g(\mathbf{y})=\pm 1$ for every $\mathbf{y}$.
The squared norm of $ \ket{ {P}_\CDD (f)}$ will become
\begin{align}\label{250623eq1}
\ep{ {P}_\CDD (f) | {P}_\CDD (f)} =
\sum_{k=0}^{l} \frac{|w_k|^2}{ {\binom{m}{k}}}  
\sum_{\substack{\mathbf{y} \in \CDD_k}}
|\widetilde{g}(\mathbf{y})|^2
= \sum_{k=0}^l |w_k|^2(1-\gamma_k),
\end{align}
and the entries of matrix $D^{(m,l,\CFF)}$ in Lemma \ref{10.6} will be
\begin{align}
D^{(m,l,\CFF)}_{k,k} =  \frac{mr}{2} \gamma_k,
\end{align}
and
\begin{align}\label{181}
D^{(m,l,\CFF)}_{k,k+1} = \frac{1}{2\sqrt{\binom{m}{k}\binom{m}{k+1}}} 
\sum_{i=1}^m 
\left| T^{(i,1,\CFF)}_{k,k+1} \right|
\tau(B,\varepsilon,i).
\end{align}

\begin{lem}\label{10.7}
For $p=2$ and $r=1$, we have 
\begin{align}
    \left\|D^{(m,l,\CFF)} - \frac{mr}{2}\mathrm{diag}(\gamma_0, \gamma_1,...,\gamma_l) \right\| \le  \tau_\infty(B,\varepsilon)(m+1)  \gamma_{\max} ,
\end{align}
where $\mathrm{diag}(\gamma_0, \gamma_1,...,\gamma_k)$ is the diagonal $(l+1) \times (l+1)$-matrix.
\end{lem}
\begin{proof}
$D^{(m,l,\CFF)} - \frac{mr}{2}\mathrm{diag}(\gamma_0, \gamma_1,...,\gamma_k)$ is a symmetric matrix whose $(k_1,k_2)$-entry is zero unless $k_1=k_2\pm 1$.
By the equation~\eqref{181}, we have 
\begin{align*}
0\le D^{(m,l,\CFF)}_{k,k+1} =& \frac{1}{2\sqrt{\binom{m}{k}\binom{m}{k+1}}} 
\sum_{i=1}^m 
\left| T^{(i,1,\CFF)}_{k,k+1} \right|
\tau(B,\varepsilon,i) \\
\le &
\frac{1}{2\sqrt{\binom{m}{k}\binom{m}{k+1}}} 
\tau_\infty(B,\varepsilon)
\sum_{i=1}^m 
\left| T^{(i,1,\CFF)}_{k,k+1} \right|.
\end{align*}
Due to the fact that  $\sum_{i=1}^m 
\left| T^{(i,1,\CFF)}_{k,k+1} \right| \le (m-k)|\CFF_k| + (k+1)|\CFF_{k+1}|$ from  the Lemma 10.7 in \cite{jordan2024optimization},
we have
\begin{align*}
D^{(m,l,\CFF)}_{k,k+1} 
\le &
\frac{1}{2\sqrt{\binom{m}{k}\binom{m}{k+1}}} 
 \tau_\infty(B,\varepsilon)
\left(
(m-k)|\CFF_k| + (k+1)|\CFF_{k+1}|
\right)\\
= & \frac{1}{2\sqrt{\binom{m}{k}\binom{m}{k+1}}} 
 \tau_\infty(B,\varepsilon)
\left(
(m-k)\gamma_k \binom{m}{k}  + (k+1)\gamma_{k+1}\binom{m}{k+1}
\right)\\
\le & \frac{1}{2} 
(\gamma_k+ \gamma_{k+1})
 \tau_\infty(B,\varepsilon) \sqrt{(k+1)(m-k)}\\
\le & \frac 12 \gamma_{\max} (m+1)  \tau_\infty(B,\varepsilon).
\end{align*}
Therefore,
we have 
\begin{align*}
    \left\|D^{(m,l,\CFF)} - \frac{mr}{2}\mathrm{diag}(\gamma_0, \gamma_1,...,\gamma_l) \right\| \le  \tau_\infty(B,\varepsilon)(m+1)  \gamma_{\max},
\end{align*}
which completes the proof.
\end{proof}

Now, we are ready to prove
Theorem~\ref{thm:main_3}.
\begin{proof}[Proof of Theorem~\ref{thm:main_3}]
Due to Theorem \ref{thm:main_2} and equation \eqref{250623eq1}, we have
\begin{align}
\Expect_{F_1,...,F_m} \left\langle s_D^{(m, l)}\right\rangle= \Expect_{F_1,...,F_m} \frac{\mathbf{w}^{\dagger} \bar{A}^{(m, l,\CDD)} \mathbf{w}}{\ep{P_\CDD(f) \big| P_\CDD(f)}} 
= \frac{\Expect_{F_1,...,F_m} \mathbf{w}^{\dagger} \bar{A}^{(m, l,\CDD)} \mathbf{w}}{\sum_{k=0}^l |w_k|^2(1-\gamma_k)}.
\end{align}
By  Lemmas \ref{10.3} and \ref{10.6}, we have
\begin{align*}
\Expect_{F_1,...,F_m}   \bar{A}^{(m, l,\CDD) }   = \frac{m }{2} I+\tau_1(B,\varepsilon) \frac{1}{2}  A^{(m, l, 0)} -  D^{(m,l,\CFF)},
\end{align*}
where $d=p-2r =0$.
By Lemma \ref{10.7}, we have
\begin{align*}
    \left\|D^{(m,l,\CFF)} - \frac{m}{2}\mathrm{diag}(\gamma_0, \gamma_1,...,\gamma_l) \right\| \le  \tau_\infty(B,\varepsilon)(m+1) \gamma_{\max}.
\end{align*}
Thus, 
\begin{align*}
\Expect_{F_1,...,F_m} \mathbf{w}^{\dagger} \bar{A}^{(m, l,\CDD)} \mathbf{w} 
\ge&  \frac{m}{2} \sum_{k=0}^l |w_k|^2(1-\gamma_k) + \frac12\tau_1(B,\varepsilon) \mathbf{w}^{\dagger} {A}^{(m, l,0)} \mathbf{w} \\
&- \tau_\infty(B,\varepsilon)(m+1) \gamma_{\max} \|\mathbf{w}\|^2. 
\end{align*}
Therefore, we have
\begin{align*}
\Expect_{F_1,...,F_m} \left\langle s_D^{(m, l)}\right\rangle 
\ge& 
\frac m2 + \frac12\tau_1(B,\varepsilon) 
\frac{\mathbf{w}^{\dagger} {A}^{(m, l,0)} \mathbf{w}}{\|  \mathbf{w}\|^2} - 
\tau_\infty(B,\varepsilon)
\frac{ (m+1)  \gamma_{\max} \|\mathbf{w}\|^2}{\sum_{k=0}^l |w_k|^2(1-\gamma_k) }
 \\
\ge & 
\frac m2 + \frac12\tau_1(B,\varepsilon) 
\frac{\mathbf{w}^{\dagger} {A}^{(m, l,0)} \mathbf{w}}{\|  \mathbf{w}\|^2} - \tau_\infty(B,\varepsilon)
\frac{ (m+1)^2 \gamma_{\max} }{1-\gamma_{\max} }.
\end{align*}
\end{proof}

\section{Several technical lemmas on  Fourier transforms}

\begin{lem}\label{250612lem1}
Let $\mathbf{y} = (y_1,...,y_k)\in \BFF_p^k$.
If $F$ is chosen uniformly randomly from all $r$-subsets of $\BFF_p$, and $y_1+\cdots + y_k \neq 0$, then we have
\begin{align}\label{250613eq1}
    \Expect_{ {F}\subseteq \BFF_p, |F|=r} 
\sum_{ \mathbf{x} = (x_1,...,x_k)\in F^k } \omega_p^{\mathbf{x} \cdot \mathbf{y} } = 
0 .
\end{align}
\end{lem}
\begin{proof}

Without loss of generality, we assume $|\mathbf{y}| = k$.
We write $x_j = x_1+w_j$ for $j=2,...,k$,
then
\begin{align*}
 &   \Expect_{ {F}\subseteq \BFF_p, |F|=r} 
\sum_{ \mathbf{x} = (x_1,...,x_k)\in F^k } \omega_p^{\mathbf{x} \cdot \mathbf{y} } \\
=& \sum_{x_1, w_2,w_3,...,w_k\in \BFF_p} \theta(w_2,...,w_k)
\omega_p^{x_1 y_1 + (x_1+w_2)y_2+\cdots + (x_1+w_k)y_k }\\
=& \sum_{  w_2,w_3,...,w_k\in \BFF_p} \theta(w_2,...,w_k)
\omega_p^{w_2y_2+\cdots + w_ky_k }
\sum_{x_1\in \BFF_p} \omega_p^{x_1 (y_1  +y_2+\cdots + y_k) }\\
=& p \delta_{y_1+\cdots + y_k =0}  \sum_{  w_2,w_3,...,w_k\in \BFF_p} \theta(w_2,...,w_k)
\omega_p^{w_2y_2+\cdots + w_ky_k }.
\end{align*}
where $\theta(w_2,...,w_k)$ is a constant determined by $w_2,...,w_k$, $p$, and $r$.
Hence, we obtain the result.
\end{proof}

In the following lemma,  we study the equation \eqref{250613eq1} under the condition $y_1+\cdots + y_k=0$.

\begin{lem}\label{250613lem2}
Let $r\ge 2$ and $\mathbf{y} = (y_1,y_2)\in (\BFF_p^*)^2$ such that  $y_1+y_2 =0$.
If $F$ is chosen uniformly randomly from all $r$-subsets of $\BFF_p$, then we have
\begin{align}
    \Expect_{ {F}\subseteq \BFF_p, |F|=r} 
\sum_{ \mathbf{x} = (x_1,x_2)\in F^2 } \omega_p^{\mathbf{x} \cdot \mathbf{y} } = r-\frac{r(r-1)}{p-1}.
\end{align}
\end{lem}

\begin{proof}
First, we can rewrite the equation as follows
\begin{align}
\Expect_{ {F}\subseteq \BFF_p, |F|=r} 
\sum_{ \mathbf{x} = (x_1,x_2)\in F^2 } \omega_p^{\mathbf{x} \cdot \mathbf{y} } 
&=r^2 \Expect_{ |F|=r} \left( 
\frac{1}{r}\Expect_{\substack{x_1,x_2\in F\\x_1= x_2}}\omega_p^{x_1(y_1+y_2)}+
\frac{r-1}{r}\Expect_{\substack{x_1,x_2\in F\\x_1\neq x_2}}\omega_p^{x_1y_1+x_2y_2}
\right).
\end{align}

Since $y_2$ is nonzero, we have
\begin{align}
\Expect_{|F|=r}\Expect_{\substack{x_1,x_2\in F\\x_1\neq x_2}}\omega_p^{x_1y_1+x_2y_2}
= \Expect_{x_1\in\BFF_p}\Expect_{w\in\BFF_p^*}\omega_p^{x_1(y_1+y_2)+wy_2}
=\frac{-1}{p-1}.
\end{align}
Hence,
\begin{align}
    \Expect_{ {F}\subseteq \BFF_p, |F|=r} 
\sum_{ \mathbf{x} = (x_1,x_2)\in F^2 } \omega_p^{\mathbf{x} \cdot \mathbf{y} } = r-\frac{r(r-1)}{p-1}.
\end{align}
\end{proof}

\begin{lem}\label{250613lem3}
Let  $r\ge 2$ and $\mathbf{y} = (y_1,y_2,y_3)\in (\BFF_p^*)^3$ such that $y_1+y_2 +y_3=0$.
If  $F$ is chosen uniformly randomly from all $r$-subsets of $\BFF_p$, we have
\begin{align}
    \Expect_{ {F}\subseteq \BFF_p, |F|=r} 
\sum_{ \mathbf{x} = (x_1,x_2,x_3)\in F^3 } \omega_p^{\mathbf{x} \cdot \mathbf{y} } 
= \frac{r(p-r)(p-2r)}{(p-1)(p-2)}.
\end{align}
\end{lem}
\begin{proof}
There are three possible cases for $x_1$, $x_2$ and $x_3$: 
(a) all are identical; (b) all are pairwise distinct; (c) they take two distinct values.
Hence,
\begin{align*}
&\Expect_{ {F}\subseteq \BFF_p, |F|=r} 
\sum_{ \mathbf{x} = (x_1,x_2,x_3)\in F^3 } \omega_p^{\mathbf{x} \cdot \mathbf{y} } \\
&= r^3 \Expect_{ |F|=r} \left( 
\frac{1}{r^2}\Expect_{\substack{x_1,x_2,x_3\in F\\x_1= x_2=x_3}}\omega_p^{x_1(y_1+y_2+y_3)} +
\frac{3(r-1)}{r^2}\Expect_{\substack{x_1,x_2,x_3\in F\\x_1\neq x_2=x_3}}\omega_p^{x_1y_1+x_2y_2+x_3y_3}
\right.
\\ &\quad+
\left.
\frac{(r-1)(r-2)}{r^2}\Expect_{\substack{x_1,x_2,x_3\in F\\x_1, x_2,x_3\text{ distinct }}}\omega_p^{x_1y_1+x_2y_2+x_3y_3}
\right)\\
&= r^3  \left( 
\frac{1}{r^2} +
\frac{3(r-1)}{r^2} 
\Expect_{\substack{w\in \BFF_p^*}}
\omega_p^{ w (y_2+y_3)} 
+
\frac{(r-1)(r-2)}{r^2} 
\Expect_{\substack{w_1,w_2\in \BFF_p^*\\w_1\neq w_2}}
\omega_p^{ w_1 y_2+w_2y_3} 
\right).
\end{align*}
Since $y_1+y_2+y_3=0$ and $y_1\neq 0$, we  have $y_2+y_3\neq 0$.
Hence,
$$\Expect_{\substack{w\in \BFF_p^*}}
\omega_p^{w (y_2+y_3)}= \frac{-1}{p-1}.$$
In addition,
\begin{align*}
\sum_{\substack{w_1,w_2\in \BFF_p^*\\w_1\neq w_2}}
\omega_p^{ w_1 y_2+w_2y_3} &= 
\sum_{\substack{w_1,w_2\in \BFF_p^*}}
\omega_p^{ w_1 y_2+w_2y_3} 
- 
\sum_{\substack{w_1,w_2\in \BFF_p^*\\w_1= w_2}}
\omega_p^{ w_1 y_2+w_2y_3} \\
&= 
\left(\sum_{\substack{w_1\in \BFF_p^*}}
\omega_p^{ w_1 y_2} \right)
\left(\sum_{\substack{w_2\in \BFF_p^*}}
\omega_p^{ w_2 y_3}\right)
- 
\sum_{\substack{w_1\in \BFF_p^*}}
\omega_p^{ w_1( y_2+y_3)} \\
&= 2.
\end{align*}
Therefore,
\begin{align*}
\Expect_{ {F}\subseteq \BFF_p, |F|=r} 
\sum_{ \mathbf{x} = (x_1,x_2,x_3)\in F^3 } \omega_p^{\mathbf{x} \cdot \mathbf{y} }  
&= r^3  \left( 
\frac{1}{r^2} -
\frac{3(r-1)}{r^2}  \frac{1}{p-1} + \frac{(r-1)(r-2)}{r^2}\frac{2}{(p-1)(p-2)}
\right)\\
&= r-\frac{3r(r-1)}{(p-1)} + \frac{2r(r-1)(r-2)}{(p-1)(p-2)}\\
&= \frac{r(p-r)(p-2r)}{(p-1)(p-2)}.
\end{align*}

\end{proof}

\end{appendix}

\section*{Acknowledgments}
K. B. is partly supported by the JobsOhio GR138220, and ARO Grant W911NF19-1-0302 and the ARO MURI Grant W911NF-20-1-0082.
D.E.K.\ is supported by the Agency for Science, Technology and Research (A*STAR) under the Central Research Fund (CRF) Award for Use-Inspired Basic Research (UIBR) and the Quantum Innovation Centre (Q.InC) Strategic Research and Translational Thrust (SRTT).

\bibliographystyle{unsrturl}
\bibliography{reference}{}

\begin{thebibliography}{10}

\bibitem{abbas2024challenges}
Amira Abbas, Andris Ambainis, Brandon Augustino, Andreas B{\"a}rtschi, Harry Buhrman, Carleton Coffrin, Giorgio Cortiana, Vedran Dunjko, Daniel~J Egger, Bruce~G Elmegreen, et~al.
\newblock Challenges and opportunities in quantum optimization.
\newblock {\em Nature Reviews Physics}, 6(12):718--735, December 2024.
\newblock \href {https://doi.org/10.1038/s42254-024-00770-9} {\path{doi:10.1038/s42254-024-00770-9}}.

\bibitem{leng2025sub}
Jiaqi Leng, Kewen Wu, Xiaodi Wu, and Yufan Zheng.
\newblock ({S}ub)exponential quantum speedup for optimization.
\newblock {\em arXiv preprint arXiv:2504.14841}, 2025.
\newblock \href {https://doi.org/10.48550/arXiv.2504.14841} {\path{doi:10.48550/arXiv.2504.14841}}.

\bibitem{pirnay2024principle}
Niklas Pirnay, Vincent Ulitzsch, Frederik Wilde, Jens Eisert, and Jean-Pierre Seifert.
\newblock An in-principle super-polynomial quantum advantage for approximating combinatorial optimization problems via computational learning theory.
\newblock {\em Science Advances}, 10(11):eadj5170, 2024.
\newblock \href {https://doi.org/10.1126/sciadv.adj5170} {\path{doi:10.1126/sciadv.adj5170}}.

\bibitem{huang2025vast}
Hsin-Yuan Huang, Soonwon Choi, Jarrod~R. McClean, and John Preskill.
\newblock The vast world of quantum advantage.
\newblock {\em arXiv preprint arXiv:2508.05720}, 2025.
\newblock \href {https://doi.org/10.48550/arXiv.2508.05720} {\path{doi:10.48550/arXiv.2508.05720}}.

\bibitem{grover1996fast}
Lov~K Grover.
\newblock A fast quantum mechanical algorithm for database search.
\newblock In {\em Proceedings of the twenty-eighth annual ACM symposium on Theory of computing}, pages 212--219, 1996.
\newblock \href {https://doi.org/10.1145/237814.237866} {\path{doi:10.1145/237814.237866}}.

\bibitem{farhi2000quantum}
Edward Farhi, Jeffrey Goldstone, Sam Gutmann, and Michael Sipser.
\newblock Quantum computation by adiabatic evolution.
\newblock {\em arXiv preprint quant-ph/0001106}, 2000.
\newblock \href {https://doi.org/10.48550/arXiv.quant-ph/0001106} {\path{doi:10.48550/arXiv.quant-ph/0001106}}.

\bibitem{albash2018adiabatic}
Tameem Albash and Daniel~A. Lidar.
\newblock Adiabatic quantum computation.
\newblock {\em Rev. Mod. Phys.}, 90:015002, Jan 2018.
\newblock \href {https://doi.org/10.1103/RevModPhys.90.015002} {\path{doi:10.1103/RevModPhys.90.015002}}.

\bibitem{moll2018quantum}
Nikolaj Moll, Panagiotis Barkoutsos, Lev~S Bishop, Jerry~M Chow, Andrew Cross, Daniel~J Egger, Stefan Filipp, Andreas Fuhrer, Jay~M Gambetta, Marc Ganzhorn, Abhinav Kandala, Antonio Mezzacapo, Peter Müller, Walter Riess, Gian Salis, John Smolin, Ivano Tavernelli, and Kristan Temme.
\newblock Quantum optimization using variational algorithms on near-term quantum devices.
\newblock {\em Quantum Science and Technology}, 3(3):030503, 2018.
\newblock \href {https://doi.org/10.1088/2058-9565/aab822} {\path{doi:10.1088/2058-9565/aab822}}.

\bibitem{cerezo2021variational}
Marco Cerezo, Andrew Arrasmith, Ryan Babbush, Simon~C Benjamin, Suguru Endo, Keisuke Fujii, Jarrod~R McClean, Kosuke Mitarai, Xiao Yuan, Lukasz Cincio, et~al.
\newblock Variational quantum algorithms.
\newblock {\em Nature Reviews Physics}, 3(9):625--644, 2021.
\newblock \href {https://doi.org/10.1038/s42254-021-00348-9} {\path{doi:10.1038/s42254-021-00348-9}}.

\bibitem{farhi2014quantum}
Edward Farhi, Jeffrey Goldstone, and Sam Gutmann.
\newblock A quantum approximate optimization algorithm.
\newblock {\em arXiv preprint arXiv:1411.4028}, 2014.
\newblock \href {https://doi.org/10.48550/arXiv.1411.4028} {\path{doi:10.48550/arXiv.1411.4028}}.

\bibitem{herrman2022multi}
Rebekah Herrman, Phillip~C. Lotshaw, James Ostrowski, Travis~S. Humble, and George Siopsis.
\newblock Multi-angle quantum approximate optimization algorithm.
\newblock {\em Scientific Reports}, 12(1):6781, Apr 2022.
\newblock \href {https://doi.org/10.1038/s41598-022-10555-8} {\path{doi:10.1038/s41598-022-10555-8}}.

\bibitem{vijendran2023expressive}
V~Vijendran, Aritra Das, Dax~Enshan Koh, Syed~M Assad, and Ping~Koy Lam.
\newblock An expressive ansatz for low-depth quantum approximate optimisation.
\newblock {\em Quantum Science and Technology}, 9(2):025010, February 2024.
\newblock \href {https://doi.org/10.1088/2058-9565/ad200a} {\path{doi:10.1088/2058-9565/ad200a}}.

\bibitem{shi2022multiangle}
Kaiyan Shi, Rebekah Herrman, Ruslan Shaydulin, Shouvanik Chakrabarti, Marco Pistoia, and Jeffrey Larson.
\newblock Multiangle {QAOA} does not always need all its angles.
\newblock In {\em 2022 IEEE/ACM 7th Symposium on Edge Computing (SEC)}, pages 414--419. IEEE, 2022.
\newblock \href {https://doi.org/10.1109/SEC54971.2022.00062} {\path{doi:10.1109/SEC54971.2022.00062}}.

\bibitem{zhao2025symmetry}
Xiumei Zhao, Yongmei Li, Guanghui Li, Yijie Shi, Sujuan Qin, and Fei Gao.
\newblock The symmetry-based expressive {QAOA} for the {MaxCut} problem.
\newblock {\em Advanced Quantum Technologies}, page 2500199, 2025.
\newblock \href {https://doi.org/10.1002/qute.202500199} {\path{doi:10.1002/qute.202500199}}.

\bibitem{pellow2024effect}
Aidan Pellow-Jarman, Shane McFarthing, Ilya Sinayskiy, Daniel~K Park, Anban Pillay, and Francesco Petruccione.
\newblock The effect of classical optimizers and ansatz depth on {QAOA} performance in noisy devices.
\newblock {\em Scientific reports}, 14(1):16011, 2024.
\newblock \href {https://doi.org/10.1038/s41598-024-66625-6} {\path{doi:10.1038/s41598-024-66625-6}}.

\bibitem{stoudenmire2024opening}
E.~M. Stoudenmire and Xavier Waintal.
\newblock Opening the black box inside {G}rover's algorithm.
\newblock {\em Phys. Rev. X}, 14:041029, Nov 2024.
\newblock \href {https://doi.org/10.1103/PhysRevX.14.041029} {\path{doi:10.1103/PhysRevX.14.041029}}.

\bibitem{farhi2011quantum}
Edward Farhi, Jeffrey Goldstone, David Gosset, Sam Gutmann, Harvey~B Meyer, and Peter Shor.
\newblock {Quantum adiabatic algorithms, small gaps, and different paths}.
\newblock {\em Quantum Information {\&} Computation}, 11(3{\&}4):181--214, 2011.
\newblock \href {https://doi.org/10.26421/QIC11.3-4-1} {\path{doi:10.26421/QIC11.3-4-1}}.

\bibitem{zhou2020quantum}
Leo Zhou, Sheng-Tao Wang, Soonwon Choi, Hannes Pichler, and Mikhail~D. Lukin.
\newblock Quantum approximate optimization algorithm: Performance, mechanism, and implementation on near-term devices.
\newblock {\em Phys. Rev. X}, 10:021067, Jun 2020.
\newblock \href {https://doi.org/10.1103/PhysRevX.10.021067} {\path{doi:10.1103/PhysRevX.10.021067}}.

\bibitem{mcclean2018barren}
Jarrod~R McClean, Sergio Boixo, Vadim~N Smelyanskiy, Ryan Babbush, and Hartmut Neven.
\newblock Barren plateaus in quantum neural network training landscapes.
\newblock {\em Nature communications}, 9(1):4812, 2018.
\newblock \href {https://doi.org/10.1038/s41467-018-07090-4} {\path{doi:10.1038/s41467-018-07090-4}}.

\bibitem{wang2021noise}
Samson Wang, Enrico Fontana, Marco Cerezo, Kunal Sharma, Akira Sone, Lukasz Cincio, and Patrick~J Coles.
\newblock Noise-induced barren plateaus in variational quantum algorithms.
\newblock {\em Nature communications}, 12(1):6961, 2021.
\newblock \href {https://doi.org/10.1038/s41467-021-27045-6} {\path{doi:10.1038/s41467-021-27045-6}}.

\bibitem{cerezo2023does}
Marco Cerezo, Martin Larocca, Diego Garc{\'\i}a-Mart{\'\i}n, Nelson~L Diaz, Paolo Braccia, Enrico Fontana, Manuel~S Rudolph, Pablo Bermejo, Aroosa Ijaz, Supanut Thanasilp, et~al.
\newblock Does provable absence of barren plateaus imply classical simulability? {O}r, why we need to rethink variational quantum computing.
\newblock {\em arXiv preprint arXiv:2312.09121}, 2023.
\newblock \href {https://doi.org/10.48550/arXiv.2312.09121} {\path{doi:10.48550/arXiv.2312.09121}}.

\bibitem{larocca2025barren}
Mart{\'i}n Larocca, Supanut Thanasilp, Samson Wang, Kunal Sharma, Jacob Biamonte, Patrick~J. Coles, Lukasz Cincio, Jarrod~R. McClean, Zo{\"e} Holmes, and M.~Cerezo.
\newblock Barren plateaus in variational quantum computing.
\newblock {\em Nature Reviews Physics}, 7(4):174--189, Apr 2025.
\newblock \href {https://doi.org/10.1038/s42254-025-00813-9} {\path{doi:10.1038/s42254-025-00813-9}}.

\bibitem{akshay2020reachability}
V.~Akshay, H.~Philathong, M.~E.~S. Morales, and J.~D. Biamonte.
\newblock Reachability deficits in quantum approximate optimization.
\newblock {\em Phys. Rev. Lett.}, 124:090504, Mar 2020.
\newblock \href {https://doi.org/10.1103/PhysRevLett.124.090504} {\path{doi:10.1103/PhysRevLett.124.090504}}.

\bibitem{bittel2021training}
Lennart Bittel and Martin Kliesch.
\newblock Training variational quantum algorithms is {NP}-hard.
\newblock {\em Phys. Rev. Lett.}, 127:120502, Sep 2021.
\newblock \href {https://doi.org/10.1103/PhysRevLett.127.120502} {\path{doi:10.1103/PhysRevLett.127.120502}}.

\bibitem{rajakumar2024trainability}
Joel Rajakumar, John Golden, Andreas B\"{a}rtschi, and Stephan Eidenbenz.
\newblock Trainability barriers in low-depth {QAOA} landscapes.
\newblock In {\em Proceedings of the 21st ACM International Conference on Computing Frontiers}, CF '24, page 199–206, New York, NY, USA, 2024. Association for Computing Machinery.
\newblock \href {https://doi.org/10.1145/3649153.3649204} {\path{doi:10.1145/3649153.3649204}}.

\bibitem{jordan2024optimization}
Stephen~P Jordan, Noah Shutty, Mary Wootters, Adam Zalcman, Alexander Schmidhuber, Robbie King, Sergei~V Isakov, Tanuj Khattar, and Ryan Babbush.
\newblock Optimization by decoded quantum interferometry.
\newblock {\em arXiv preprint arXiv:2408.08292}, 2024.
\newblock \href {https://doi.org/10.48550/arXiv.2408.08292} {\path{doi:10.48550/arXiv.2408.08292}}.

\bibitem{patamawisut2025quantum}
Natchapol Patamawisut, Naphan Benchasattabuse, Michal Hajdu{\v{s}}ek, and Rodney Van~Meter.
\newblock Quantum circuit design for decoded quantum interferometry.
\newblock {\em arXiv preprint arXiv:2504.18334}, 2025.
\newblock \href {https://doi.org/10.48550/arXiv.2504.18334} {\path{doi:10.48550/arXiv.2504.18334}}.

\bibitem{chailloux2025quantum}
Andr\'{e} Chailloux and Jean-Pierre Tillich.
\newblock Quantum advantage from soft decoders.
\newblock In {\em Proceedings of the 57th Annual ACM Symposium on Theory of Computing}, STOC '25, page 738–749, New York, NY, USA, 2025. Association for Computing Machinery.
\newblock \href {https://doi.org/10.1145/3717823.3718319} {\path{doi:10.1145/3717823.3718319}}.

\bibitem{ralli2025bridging}
Alexis Ralli, Tim Weaving, Peter~V Coveney, and Peter~J Love.
\newblock Bridging quantum chemistry and {MaxCut}: Classical performance guarantees and quantum algorithms for the {H}artree-{F}ock method.
\newblock {\em arXiv preprint arXiv:2506.04223}, 2025.
\newblock \href {https://doi.org/10.48550/arXiv.2506.04223} {\path{doi:10.48550/arXiv.2506.04223}}.

\bibitem{preskill2018quantum}
John Preskill.
\newblock Quantum {C}omputing in the {NISQ} era and beyond.
\newblock {\em {Quantum}}, 2:79, August 2018.
\newblock \href {https://doi.org/10.22331/q-2018-08-06-79} {\path{doi:10.22331/q-2018-08-06-79}}.

\bibitem{cheng2023noisy}
Bin Cheng, Xiu-Hao Deng, Xiu Gu, Yu~He, Guangchong Hu, Peihao Huang, Jun Li, Ben-Chuan Lin, Dawei Lu, Yao Lu, Chudan Qiu, Hui Wang, Tao Xin, Shi Yu, Man-Hong Yung, Junkai Zeng, Song Zhang, Youpeng Zhong, Xinhua Peng, Franco Nori, and Dapeng Yu.
\newblock Noisy intermediate-scale quantum computers.
\newblock {\em Frontiers of Physics}, 18(2):21308, Mar 2023.
\newblock \href {https://doi.org/10.1007/s11467-022-1249-z} {\path{doi:10.1007/s11467-022-1249-z}}.

\bibitem{preskill2025beyond}
John Preskill.
\newblock Beyond {NISQ}: The megaquop machine.
\newblock {\em ACM Transactions on Quantum Computing}, 6(3), April 2025.
\newblock \href {https://doi.org/10.1145/3723153} {\path{doi:10.1145/3723153}}.

\bibitem{koh2022classical}
Dax~Enshan Koh and Sabee Grewal.
\newblock Classical {S}hadows {W}ith {N}oise.
\newblock {\em {Quantum}}, 6:776, August 2022.
\newblock \href {https://doi.org/10.22331/q-2022-08-16-776} {\path{doi:10.22331/q-2022-08-16-776}}.

\bibitem{chen2020robust}
Senrui Chen, Wenjun Yu, Pei Zeng, and Steven~T. Flammia.
\newblock {Robust Shadow Estimation}.
\newblock {\em PRX Quantum}, 2:030348, Sep 2021.
\newblock \href {https://doi.org/10.1103/PRXQuantum.2.030348} {\path{doi:10.1103/PRXQuantum.2.030348}}.

\bibitem{Bunpj22}
Kaifeng Bu, Dax~Enshan Koh, Roy Garcia, and Arthur Jaffe.
\newblock Classical shadows with {P}auli-invariant unitary ensembles.
\newblock {\em npj Quantum Information}, 10(1):6, Jan 2024.
\newblock \href {https://doi.org/10.1038/s41534-023-00801-w} {\path{doi:10.1038/s41534-023-00801-w}}.

\bibitem{wu2024error-mitigated}
Bujiao Wu and Dax~Enshan Koh.
\newblock Error-mitigated fermionic classical shadows on noisy quantum devices.
\newblock {\em npj Quantum Information}, 10(1):39, 2024.
\newblock \href {https://doi.org/10.1038/s41534-024-00836-7} {\path{doi:10.1038/s41534-024-00836-7}}.

\bibitem{flammia2020efficient}
Steven~T. Flammia and Joel~J. Wallman.
\newblock Efficient estimation of {P}auli channels.
\newblock {\em ACM Transactions on Quantum Computing}, 1(1):3, December 2020.
\newblock \href {https://doi.org/10.1145/3408039} {\path{doi:10.1145/3408039}}.

\bibitem{chen2022quantum}
Senrui Chen, Sisi Zhou, Alireza Seif, and Liang Jiang.
\newblock Quantum advantages for {P}auli channel estimation.
\newblock {\em Phys. Rev. A}, 105:032435, Mar 2022.
\newblock \href {https://doi.org/10.1103/PhysRevA.105.032435} {\path{doi:10.1103/PhysRevA.105.032435}}.

\bibitem{montanaro2010quantum}
Ashley Montanaro and Tobias~J. Osborne.
\newblock Quantum boolean functions.
\newblock {\em Chicago Journal of Theoretical Computer Science}, 2010(1), January 2010.
\newblock \href {https://doi.org/10.4086/cjtcs.2010.001} {\path{doi:10.4086/cjtcs.2010.001}}.

\bibitem{BD19}
Kaifeng Bu and Dax~Enshan Koh.
\newblock Efficient classical simulation of {C}lifford circuits with nonstabilizer input states.
\newblock {\em Phys. Rev. Lett.}, 123:170502, Oct 2019.
\newblock \href {https://doi.org/10.1103/PhysRevLett.123.170502} {\path{doi:10.1103/PhysRevLett.123.170502}}.

\bibitem{Bucomplexity22}
Kaifeng Bu, Roy~J. Garcia, Arthur Jaffe, Dax~Enshan Koh, and Lu~Li.
\newblock Complexity of quantum circuits via sensitivity, magic, and coherence.
\newblock {\em Communications in Mathematical Physics}, 405(7):161, Jun 2024.
\newblock \href {https://doi.org/10.1007/s00220-024-05030-6} {\path{doi:10.1007/s00220-024-05030-6}}.

\bibitem{BGJ23a}
Kaifeng Bu, Weichen Gu, and Arthur Jaffe.
\newblock Quantum entropy and central limit theorem.
\newblock {\em Proceedings of the National Academy of Sciences}, 120(25):e2304589120, 2023.
\newblock \href {https://doi.org/10.1073/pnas.2304589120} {\path{doi:10.1073/pnas.2304589120}}.

\bibitem{BGJ23b}
Kaifeng Bu, Weichen Gu, and Arthur Jaffe.
\newblock Discrete quantum {G}aussians and central limit theorem.
\newblock {\em arXiv preprint arXiv:2302.08423}, 2023.
\newblock \href {https://doi.org/10.48550/arXiv.2302.08423} {\path{doi:10.48550/arXiv.2302.08423}}.

\bibitem{BGJ23c}
Kaifeng Bu, Weichen Gu, and Arthur Jaffe.
\newblock Stabilizer testing and magic entropy via quantum fourier analysis.
\newblock {\em arXiv preprint arXiv:2306.09292}, 2023.
\newblock \href {https://doi.org/10.48550/arXiv.2306.09292} {\path{doi:10.48550/arXiv.2306.09292}}.

\bibitem{Bu2025quantum}
Kaifeng Bu, Weichen Gu, and Arthur Jaffe.
\newblock Quantum {R}uzsa divergence to quantify magic.
\newblock {\em IEEE Transactions on Information Theory}, 71(4):2726--2740, 2025.
\newblock \href {https://doi.org/10.1109/TIT.2025.3543276} {\path{doi:10.1109/TIT.2025.3543276}}.

\bibitem{BJPRL25}
Kaifeng Bu and Arthur Jaffe.
\newblock Magic resource can enhance the quantum capacity of channels.
\newblock {\em Phys. Rev. Lett.}, 134:050202, Feb 2025.
\newblock \href {https://doi.org/10.1103/PhysRevLett.134.050202} {\path{doi:10.1103/PhysRevLett.134.050202}}.

\bibitem{xue2021effects}
Cheng Xue, Zhao-Yun Chen, Yu-Chun Wu, and Guo-Ping Guo.
\newblock Effects of quantum noise on quantum approximate optimization algorithm.
\newblock {\em Chinese Physics Letters}, 38(3):030302, mar 2021.
\newblock \href {https://doi.org/10.1088/0256-307X/38/3/030302} {\path{doi:10.1088/0256-307X/38/3/030302}}.

\bibitem{marshall2020characterizing}
Jeffrey Marshall, Filip Wudarski, Stuart Hadfield, and Tad Hogg.
\newblock Characterizing local noise in {QAOA} circuits.
\newblock {\em IOP SciNotes}, 1(2):025208, 2020.
\newblock \href {https://doi.org/10.1088/2633-1357/abb0d7} {\path{doi:10.1088/2633-1357/abb0d7}}.

\end{thebibliography}
\end{document}